\documentclass[11pt]{article}
\usepackage[usenames,dvipsnames,svgnames,table]{xcolor}
\usepackage{enumitem}
\usepackage{graphicx}
\usepackage{fancybox}
\usepackage{comment}
\usepackage{xcolor}
\usepackage{nameref}
\usepackage{bbm}

\definecolor{ForestGreen}{rgb}{0.1333,0.5451,0.1333}
\definecolor{DarkRed}{rgb}{0.8,0,0}
\definecolor{Red}{rgb}{1,0,0}
\usepackage[linktocpage=true,
pagebackref=true,colorlinks,
linkcolor=ForestGreen,citecolor=ForestGreen,
bookmarks,bookmarksopen,bookmarksnumbered]
{hyperref}
\usepackage[nottoc,numbib]{tocbibind}

\usepackage{amsmath}
\usepackage{amssymb}
\usepackage{amsthm}

\usepackage[ruled, noend, linesnumbered]{algorithm2e}

\usepackage{subfig}

\usepackage{thm-restate}

\usepackage{float}
\usepackage{cleveref}

\newtheorem{theorem}{Theorem}[section]

\newtheorem{corollary}[theorem]{Corollary}
\newtheorem{lemma}[theorem]{Lemma}

\newtheorem{claim}[theorem]{Claim}

\newtheorem{invariant}[theorem]{Invariant}

\newtheorem{definition}[theorem]{Definition}
\newtheorem{remark}[theorem]{Remark}

\newtheorem*{theorem*}{Theorem}
\newtheorem*{corollary*}{Corollary}
\newtheorem*{conjecture*}{Conjecture}
\newtheorem*{lemma*}{Lemma}
\newtheorem*{thm*}{Theorem}
\newtheorem*{prop*}{Proposition}
\newtheorem*{obs*}{Observation}
\newtheorem*{definition*}{Definition}

\newtheorem*{remark*}{Remark}
\newtheorem*{rec*}{Recommendation}

\newenvironment{fminipage}%
  {\begin{Sbox}\begin{minipage}}%
  {\end{minipage}\end{Sbox}\fbox{\TheSbox}}

\def\defeq{\stackrel{\mathrm{def}}{=}}

\newcommand{\veryshortarrow}[1][3pt]{\mathrel{%
   \hbox{\rule[\dimexpr\fontdimen22\textfont2-.2pt\relax]{#1}{.4pt}}%
   \mkern-4mu\hbox{\usefont{U}{lasy}{m}{n}\symbol{41}}}}

\def\norm#1{\left\| #1 \right\|_{\veryshortarrow}}

\newcommand\PPi{\boldsymbol{\Pi}}

\renewcommand\ll{\boldsymbol{\mathit{l}}}

\renewcommand{\deg}{\operatorname{deg}}

\newcommand\Otil{\widetilde{O}}

\newcommand\R{\mathbb{R}}

\newcommand{\proj}{\PPi}

\newcommand{\concat}{\oplus}

\newcommand\Z{\mathbb{Z}}

\renewcommand{\forall}{\mathrm{\text{ for all }}}

\newcommand{\length}{\mathrm{length}}

\newcommand{\pconcat}{\oplus}

\renewcommand{\root}{\mathsf{root}}

\renewcommand{\hat}{\widehat}
\renewcommand{\tilde}{\widetilde}

\DeclareFontFamily{U}{mathb}{\hyphenchar\font45}
\DeclareFontShape{U}{mathb}{m}{n}{<5> <6> <7> <8> <9> <10> gen * mathb
<10.95> mathb10 <12> <14.4> <17.28> <20.74> <24.88> mathb12}{}
\DeclareSymbolFont{mathb}{U}{mathb}{m}{n}
\DeclareMathSymbol{\rcirclearrow}{\mathbin}{mathb}{'367}

\renewcommand{\bar}{\overline}

\newif\ifrandom
\randomtrue

\newcommand{\str}{{\mathsf{str}}}

\newcommand{\todolater}[1]{}

\DeclareUnicodeCharacter{2113}{$\ell$}

\renewcommand{\root}{\mathsf{root}}

\usepackage{fullpage}

\DeclareMathOperator{\vcong}{vcong}
\DeclareMathOperator{\econg}{econg}
\DeclareMathOperator{\master}{master}

\DeclareMathOperator{\dist}{dist}

\title{Bootstrapping Dynamic APSP via Sparsification}
\newcommand*\samethanks[1][\value{footnote}]{\footnotemark[#1]}

\author{Rasmus Kyng\thanks{The research leading to these results has received funding from grant no. 200021 204787 of the Swiss National Science Foundation.} \\ ETH Zurich \\ kyng@inf.ethz.ch\and Simon Meierhans\samethanks[1] \\ ETH Zurich \\ mesimon@inf.ethz.ch \and Gernot Zöcklein \\ ETH Zurich \\ gzoecklein@ethz.ch}

\begin{document}

\maketitle

\begin{abstract}

    We give a simple algorithm for the dynamic approximate All-Pairs Shortest Paths (APSP) problem. Given a graph $G = (V, E, \ll)$ with polynomially bounded edge lengths, our data structure processes $|E|$ edge insertions and deletions in total time $|E|^{1 + o(1)}$ and provides query access to $|E|^{o(1)}$-approximate distances in time $\Otil(1)$ per query. 

    We produce a data structure that mimics Thorup-Zwick distance oracles \cite{thorup05}, but is dynamic and deterministic.
    Our algorithm selects a small number of pivot vertices.
    Then, for every other vertex, it reduces distance computation to maintaining distances to a small neighborhood around that vertex
    and to the nearest pivot.
    We maintain distances between pivots efficiently by representing them in a smaller graph and recursing.
    We construct these smaller graphs by 
    (a) reducing vertex count using the dynamic distance-preserving core graphs of Kyng-Meierhans-Probst Gutenberg \cite{kyng2023dynamic} in a black-box manner 
    and 
    (b) reducing edge-count using a dynamic spanner akin to Chen-Kyng-Liu-Meierhans-Probst Gutenberg \cite{CKL24}.
    Our dynamic spanner internally uses an APSP data structure. 
    Choosing a large enough size reduction factor in the first step allows us to simultaneously bootstrap 
    our spanner and a dynamic APSP data structure.
    Notably, our approach does not need expander graphs, an otherwise ubiquitous tool in derandomization.

\end{abstract}

\newpage

\section{Introduction}

The dynamic \emph{All-Pairs Shortest Paths} (APSP) problem asks to maintain query access to pairwise distances for a dynamic graph $G = (V, E)$. Recently, the setting of maintaining crude distance estimates at low computational cost has received significant attention \cite{HKN18, chechik18, FGH21, CS21a, chuzhoy21, BGS21, forster_et_al:LIPIcs.ESA.2023.50, forster23incremental} culminating in deterministic algorithms with sub-polynomial update and query time for graphs receiving both edge insertions and deletions \cite{CZ23, kyng2023dynamic, haeupler2024dynamicdeterministicconstantapproximatedistance}. 
In contrast to other approaches, \cite{kyng2023dynamic} allows efficient implicit access to approximate shortest paths via a small collection of dynamic trees.
This is crucial for applications to decremental single source shortest paths and incremental minimum cost flows \cite{CKL24}.

All existing dynamic approximate APSP data structures that work against adaptive adversaries are highly complex \cite{CZ23,kyng2023dynamic,haeupler2024dynamicdeterministicconstantapproximatedistance}.
The algorithm \cite{kyng2023dynamic} establishes the following
\emph{Thorup-Zwick-distance oracle}-inspired framework:
To maintain approximate APSP in the graph, we select a small (dynamic) set of pivot vertices.
For every other vertex, we can reduce distance computation to (a) computing distances in a small neighborhood around the vertex and (b) computing the distance to the nearest pivot vertex.
Then,
\cite{kyng2023dynamic} maintains 
a smaller graph where distances between the pivot vertices can be computed, and recursively maintains these with the same approach.
To maintain this smaller graph, \cite{kyng2023dynamic} designs a dynamic vertex sparsification procedure and performs edge sparsification using the dynamic spanner of \cite{maxflow,detmax}.
This spanner internally relies on complex data structures for decremental APSP on expanders \cite{CS21a} and uses expander embedding techniques in a non-trivial way.
We will use the overall dynamic Thorup-Zwick approach of \cite{kyng2023dynamic}, and we will maintain a small graph for computing distances between pivots in a simpler way.

A key ingredient for us is a new, more powerful fully-dynamic spanner from \cite{CKL24}. This spanner requires both dynamic APSP and expander embedding techniques to maintain.
In this work, we observe that this spanner can be simplified so that it only requires dynamic APSP to maintain.
Thus, given a dynamic APSP data structure, we can maintain a fully-dynamic spanner. 
But, \cite{kyng2023dynamic} also showed that given a fully-dynamic spanner, we can obtain dynamic APSP.
This presents a striking opportunity: we can simultaneously bootstrap a dynamic APSP data structure and spanner. 
This gives our main result, a dynamic approximate APSP data structure.

\begin{restatable}{theorem}{mainTheorem}
    \label{thm:main}
    For a graph $G = (V, E, \ll)$ with polynomially bounded integral edge lengths $\ll \in \R^E_{> 0}$ undergoing up to $m$ edge updates (insertions and deletions), there is an algorithm with total update time $m^{1 + o(1)}$ that maintains query access to approximate distances $\widetilde{\dist}(u, v)$ such that
    \begin{align*}
        \dist(u, v) \leq \widetilde{\dist}(u, v) \leq m^{o(1)} \cdot \dist(u, v). 
    \end{align*}
    The query time is $\Otil(1)$.
\end{restatable}

While we present our theorem in the setting of processing $m$ updates in time $m^{1 + o(1)}$ to simplify the presentation, we believe that this algorithm can be extended directly to the worst-case update time setting of \cite{kyng2023dynamic}. Furthermore, the witness paths can be implicitly represented as forest paths as in \cite{kyng2023dynamic}, which enables the cheap flow maintenance via dynamic tree data structures necessary for applications. 

\paragraph{Roadmap. } We first define the necessary preliminaries in \Cref{sec:preliminaries}. Then, we present a simplified version of our algorithm as a warm-up in \Cref{sec:warm_up}. Afterwards, we present our bootstrapped algorithm based on (a) the KMG-vertex sparsifier \cite{kyng2023dynamic} and (b) a dynamic spanner in \Cref{sec:algorithm}. Then, we present details of (b), our dynamic spanner, in \Cref{sec:spanner}. Finally, we sketch the construction of (a), the KMG-vertex sparsifier, in \Cref{apx:vertex_sparsifier}.

\section{Preliminaries} \label{sec:preliminaries}

\paragraph{Graph  Notation. } We use $G = (V, E, \ll)$ to refer to a graph $G$ with vertex set $V$, edge set $E$ and length vector $\ll \in \R^E_{\geq 1}$. We will restrict our attention to polynomially bounded integral lengths, and will denote their upper bound by $L$. 

\paragraph{Paths, Distances and Balls. } For two vertices $u, v$, we let $\pi^G_{u, v}$ denote the shortest path from $u$ to $v$ in $G$ (lowest sum of edge lengths). Furthermore, we let $\dist_G(u, v)$  be the length of said path. 

We let $B_G(u, v) \defeq \{w: \dist(u, w) < \dist(u, v)\}$ and $\bar{B}(u, v) \defeq \{w: \dist(u, w) \leq \dist(u, v)\}$. Then, for a set $A \subseteq V$ we let $B_G(u, A) \defeq \bigcup_{a \in A}B_G(u, a)$ and $\bar{B}_G(u, A) \defeq \bigcup_{a \in A} \bar{B}_G(u, a)$ respectively. 

Finally, we define the neighborhood of a vertex $u$ as $\mathcal{N}(u) \defeq \{v: (u, v) \in E\} \cup \{u\}$ and of a set $A \subseteq V$ as $\mathcal{N}(A) \defeq \bigcup_{a \in A} \mathcal{N}(a)$. For a vertex $v$, we let $E_G(v) \defeq \{(v, u): u \in \mathcal{N}(v)\} \subseteq E(G)$.

When it is clear from the context, we sometimes omit the subscript/superscript $G$.  

\begin{definition}[Spanner]
    Given a graph $G = (V, E, \ll)$, we say that a subgraph $H \subseteq G$ is an $\alpha$-spanner of $G$ if for any $u, v \in V$ we have $\dist_G(u, v) \leq \dist_H(u, v) \leq \alpha \cdot \dist_G(u, v)$.
\end{definition}

\paragraph{Dynamic Graphs. }Vertex splits are a crucial operation in our algorithms. They arise since we contract graph components and represent them as vertices. When such a component needs to be divided, this naturally corresponds to a vertex split defined as follows. 

\begin{definition}[Vertex Split]\label{def:vertex_split}
    For a graph $G = (V, E, \ll)$, a vertex split takes a vertex $v \in V$ and a set $U \subseteq \mathcal{N}(v)$, and replaces the vertex $v$ with two vertices $v', v''$ such that $\mathcal{N}(v') = U$ and  $\mathcal{N}(v'') = \mathcal{N}(v) \setminus U$.
\end{definition}

\begin{definition}[Master Nodes]Initially, each $v \in V(G^{(0)})$ is associated with a unique master node $\master(v) = \{v\}$. Then, at any time $t$, if some $v \in V(G^{(t)})$ is split into $v'$ and $v''$, we set $\master(v') = \master(v'') = \{v\} \cup \master(v)$. 
\end{definition}

We can also define our notion of fully-dynamic graphs, which includes vertex splits as an operation. 

\begin{definition}[fully-dynamic Graph]
A fully-dynamic graph is a graph undergoing edge insertions/deletions, vertex splits and isolated vertex insertions.
\end{definition}

Whenever we disallow vertex splits, we refer to edge-dynamic graphs instead. 

\begin{definition}[Edge-Dynamic Graph]
    An edge-dynamic graph is a graph undergoing edge insertions/deletions and isolated vertex insertions.
\end{definition}

\paragraph{Dynamic Objects and Recourse. }
For a dynamic object $(F(t))_{t \in [T]}$ defined on an index set $[T] = \{1,\ldots, T\}$, we refer to $F^{(T')} \defeq (F(t))_{t \in [T']}$ as the full sequence up to time $T'$. We can then apply functions $g(\cdot)$ that operate on sequences to  $F^{(T')}$, e.g. $g(F^{(T')})$. This is for example very useful to keep track of how an object changes over time. By providing a difference function $\triangle(F(t), F(t-1))$ for a dynamic object $F$, one may define the recourse as the sum of differences $\norm{F^{(t)}} \defeq \sum_{i = 0}^{t - 1} \triangle(F(i + 1), F(i))$.

For functions $f$ that operate on a static object $G$ instead of on a sequence, we overload the notation by saying $f(G^{(t)})$ is simply $f$ applied to the last object in the sequence, that is, $f(G^{(t)}) = f(G(t))$. For example, $\Delta_{\max}(G^{(t)})$ is defined to be the maximum degree of the graph $G$ at time $t$.

Whenever we make a statements about a dynamic object $F$ without making reference to a specific time $t$, we implicitly mean that the statement holds at all times.
For example, $\norm{H} \leq \norm{G}$ formally means
$\forall t: \norm{H^{(t)}} \leq \norm{G^{(t)}}$.
Meanwhile, and $|E(H)| \leq |E(G)|$
formally means $\forall t: |E(H^{(t)})| \leq |E(G^{(t)})|$, and by our convention for applying functions of static objects to a dynamic object, we have that $|E(H^{(t)})| \leq |E(G^{(t)})|$ means  $|E(H(t))| \leq |E(G(t))|$.

The time scales we use are fine-grained as we will consider multiple dynamic objects that change more frequently than the input graph $G$ does. 
We will choose an index set for the dynamic sequence which is fine-grained enough to capture changes to all the objects relevant to a given proof.
A dynamic object $H$ might not change in a time span in which another object $F$ changes multiple times. In this case, to have sequences that are indexed by the same set, we simply pad the sequence of $H$ with the same object, so that both dynamic objects would share the same index set $T$.

For fully-dynamic graphs $G$, we let $\triangle(G(t), G(t-1))$ denote the number of edge insertions, edge deletions, vertex splits and isolated vertex insertions that happened between $G(t)$ and $G(t - 1)$.
Note that generally, vertex splits cannot be implemented in constant time, and thus recourse bounds are not directly relevant to bounding running time.
Nevertheless, this convention for measuring recourse is useful, as it is tailored precisely to the type of recourse bounds we prove.

We can store anything that occurs at various times throughout our algorithm in a dynamic object and define what recourse means for it. For example, we will repeatedly call a procedure $\textsc{ReduceDegree}(X)$. If at time $t$ we call the function with input $X(t)$, we might want to track the quantity $\sum_{t' \leq t} |X(t')|$, so we simply let $R_X$ be the dynamic object such that $R_X(t) = \sum_{t' \leq t} |X(t')|$. By defining the difference $\Delta(R_X(t), R_X(t-1)) = R_X(t) - R_X(t-1)$, we can find out how much the value changed during two points in time. A statement such as $\norm{H} \leq \norm{G} + \norm{R_X}$ then means that at any time $t$ (suitably fine-grained), the total amount of changes that $H$ underwent are not more than the amount of $G$, plus some extra number of updates determined by the calls to $\textsc{ReduceDegree}(\cdot)$ made up to that time.

\paragraph{Dynamic APSP. } We introduce APSP data structures for edge-dynamic graphs.
\begin{definition}[Dynamic APSP] \label{def:apsp}
    For an edge-dynamic graph $G = (V, E, \ll)$, a $\gamma$-approximate $\beta$-query APSP data structure supports the following operations.
    \begin{itemize}
        \item $\textsc{AddEdge}(u,v)$ / $\textsc{RemoveEdge}(u, v)$: Add/Remove edge $(u, v)$ from/to $G$.
        \item $\textsc{Dist}(u, v)$: Returns in time $\beta$ a distance estimate $\widetilde{\dist}(u,v)$ such that at all times
        \begin{equation*}
            \dist_G(u, v) \leq \widetilde{\dist}(u,v) \leq \gamma \cdot \dist_G(u,v). 
        \end{equation*}
        \item $\textsc{Path}(u,v)$: Returns a $uv$-path $P$ in $G$ of length $|P| \leq \textsc{Dist}(u,v)$ in time $\beta \cdot |P|$.
    \end{itemize}
    Given an APSP data structure, we let $APSP(|E|, \Delta_{\max}(G), u)$ denote the total run time of initialization and $u$ calls to $\textsc{AddEdge}$ and $\textsc{RemoveEdge}$ on an edge-dynamic graph with $\Delta_{\max}(G) \leq \Delta$ and $|E|$ initial edges.

    We sometimes also call $\gamma$ the (worst-case) stretch of the data structure.
\end{definition}

\paragraph{Degree Reduction Through Binary Search Trees}\label{fact:BSTred}Given an \emph{edge-dynamic} graph $G$, we can maintain in time $\Otil(1) \cdot \norm{G}$ another edge-dynamic graph $G'$ in which we replace vertices by Binary Search Trees with one leaf node per edge of the vertex in the original graph.
It is immediate that if we can find shortest paths in $G'$, we can translate them into shortest paths in $G$. The advantage of this technique is that while 
$|E(G')| = O(|E(G)|)$, $|V(G')| = O(|E(G)|)$ and $\norm{G'} = \Otil(\norm{G})$, we have $\Delta_{\max}(G') \leq 3$.

\section{Warm up: Low-Recourse Thorup-Zwick'05} \label{sec:warm_up}

Before we describe our algorithm, we turn our attention to a much more modest goal that still showcases the main ideas used in our full construction: maintaining a relaxed version of Thorup-Zwick pivots \cite{thorup05} in a fully-dynamic graph with low recourse. This allows us to do away with the careful controlling of maximum degrees in intermediate graphs, which is the root cause of most of the technicalities when making the construction computationally efficient. Furthermore, we do not need to bootstrap our data structure in this setting. 

\paragraph{Pivot Hierarchies.} We first introduce pivot hierarchies. 

\begin{definition}[Pivot Hierarchy] \label{def:pivot_hierarchy}
    For a graph $G = (V, E, \ll)$, we let a pivot hierarchy of depth $\Lambda$ and cluster size $\gamma$ be  
    \begin{itemize}
        \item a collection of vertex sets $\mathcal{C} = A_0, \ldots, A_{\Lambda}$ with $A_i \subseteq V$ for all $i$, $A_0 = V$ and $|A_{\Lambda}| = 1$,
        \item pivot maps $p_i: A_{i} \rightarrow A_{i + 1}$ for all $0 \leq i < \Lambda$ and
        \item cluster maps $C_i: A_i \rightarrow 2^{A_i}$ for all $0 \leq i < \Lambda$ such that for all $u \in A_i$ $|\{v \in A_i: v \in C_i(u)\}| \leq \gamma$ and $u \in C_i(u) $. 
    \end{itemize}
\end{definition}

Pivot maps should be thought of as mapping vertices to smaller and smaller sets of representatives, and the cluster maps should be thought of as small neighborhoods of pivots for which (approximate) distances are stored. In classical Thorup-Zwick, the cluster $C_i(v)$ consists of all vertices in $A_i$ that are closer to $v$ than $p_i(v)$. In our algorithms, these balls become somewhat distorted due to sparsification. We then define the approximate distances maintained by the hierarchy.

\begin{definition}[Distances]
    For a pivot hierarchy of depth $\Lambda$, we define $\widetilde{\dist}^{(\Lambda)}(u, v) \defeq 0$ for $u, v \in A_{\Lambda}$ and 
    \begin{equation*}
        \widetilde{\dist}^{(i)}(u, v) \defeq \begin{cases}
            \dist_{G}(u,v) & \text{if } u \in C_{i}(v) \\
            \dist_{G}(u, p_i(u)) + \dist^{(i + 1)}(p_i(u), p_i(v)) + \dist_{G}(p_i(v), v) & \text{otherwise}
        \end{cases}.
    \end{equation*}
    We let $\widetilde{\dist}(u, v) \defeq \widetilde{\dist}^{(0)}(u, v)$ for $u, v \in V$ and call a pivot hierarchy $\alpha$ distance preserving if $\widetilde{\dist}(u, v) \leq \alpha \cdot \dist_G(u, v)$ for all $u,v \in V$. 
\end{definition}

We note that efficient distance queries can be implemented as long as the distances within clusters and the distances to pivots are efficiently maintained. In this warm up, we only focus on maintaining the collection $\mathcal{C}$ with low recourse. However, our full algorithm will maintain this additional information. 

Formally, the goal of this warm-up section is to prove the following theorem.

\begin{theorem}\label{thm:low_recourse_pivots}
    For an edge-dynamic graph $G = (V, E)$ with $n = |V|$ vertices, a $\sqrt{\log n}$-depth $n^{o(1)}$-cluster size pivot hierarchy can be maintained in polynomial time per update with total recourse $\norm{\mathcal{C}} \leq n^{o(1)} \cdot \norm{G}$ on the pivot sets.
\end{theorem}

We emphasize that this theorem is not particularly useful, as it comes without meaningful guarantees on the update time, and one could just use Dijkstras algorithm to obtain exact distances in this regime.
Even recourse guarantee is also of limited use, as we do not control the recourse of the pivot maps or the cluster maps of the pivot hierarchy.
Nevertheless, our algorithm for proving this theorem is instructive and closely mirrors our final dynamic APSP algorithm.

The illustrative algorithm for maintaining pivot hierarchies repeatedly decreases the vertex count via low-recourse core graphs, and the edge count via a low-recourse spanner. We present these two pieces separately. 

\paragraph{Low-Recourse KMG-vertex sparsifier. } In this paragraph, we state a theorem about maintaining a KMG-vertex sparsifier with low recourse \cite{kyng2023dynamic}. Although their algorithm only works for edge-dynamic graphs, it can be extended to fully-dynamic graphs rather straightforwardly in the recourse setting. Recall that we count a vertex split as a single operation when defining the recourse of a fully-dynamic graph. 

\begin{definition}[Vertex Sparsifier]%
    Given $G = (V, E, \ll)$ and $A \subseteq V$, we call a graph $\Tilde{G}$ with vertex set $V(\Tilde{G}) \supseteq A$ a $\beta$-approximate vertex sparsifier of $G$ with respect to the set $A$ if for every two vertices $u, v \in V(\Tilde{G})$ we have $\dist_G(u, v) \leq \dist_{\Tilde{G}}(u, v)$ and if $u, v \in A$ we further have $\dist_{\Tilde{G}}(u, v) \leq \beta \cdot \dist_G(u, v)$.
\end{definition}

\begin{restatable}[See Theorem 3.1 in \cite{kyng2023dynamic} \& \Cref{apx:vertex_sparsifier}]{theorem}{vertexSparsifier}\label{thm:RecourseVertexSparsifier}
Consider a size reduction parameter $k > 1$ and an fully-dynamic graph $G$ with lengths in $[1, L]$. Then, there is a deterministic algorithm that explicitly maintains, 
\begin{enumerate}
    \item a monotonically growing pivot set $A \subseteq V(G)$ and pivot function $p : V \rightarrow A$, such that  $|A| \leq n/k + 2\norm{G}$ and $|C(u)| \leq \tilde{O}(k)$ for all $u \in V$ where $C(u) = \{v \in V: \dist_G(u, v) < \dist_G(u, p(u)\}$.
    \item \label{prop:workhorseDistancePreserve} a fully-dynamic graph $\Tilde{G}$, such that $\Tilde{G}$ is a $\gamma_{\textit{VS}}$-approximate vertex sparsifier of $G$ with respect to $A$. We have that $\Tilde{G}^{(0)}$ has at most $m \cdot \gamma_{\textit{VS}}$ edges and $\gamma_{\textit{VS}} \cdot m /k$ vertices, where $\gamma_{\textit{VS}} = \Otil(1)$.
    The following properties hold:
    \begin{enumerate}
        \item \label{prop:maxDegreeVS}  $\Tilde{G}$ has lengths $\ll_{\Tilde{G}}$ in $[1, nL]$, and 
        \item given any edge $e = (u,v) \in \Tilde{G}$, there exists a $uv$-path $P$ in $G$ with $l_G(P) \leq l_{\Tilde{G}}(e)$. 
        \item  $\norm{\Tilde{G}} \leq \gamma_{\text{VS}} \cdot (\norm{G} + m)$ and $\norm{D_{\tilde{G}}} \leq \gamma_{\text{VS}} \cdot \norm{G}$ where $D_{\tilde{G}}$ is the dynamic object containing the decremental operations (deletions and vertex splits) on $\tilde{G}$.
    \end{enumerate}
\end{enumerate}
The algorithm runs in polynomial time.
\end{restatable}

\begin{remark}
    We always associate a map from $V(\tilde{G})$ to $V$ with a vertex sparsifier, where multiple vertices in $V(\tilde{G})$ might map to the same vertex in $V$.
    When we apply multiple vertex sparsifiers to a graph $G$, we refer to recursively applying this map as mapping vertices back to $G$. 
\end{remark}

\paragraph{Low-Recourse Dynamic Spanner. } In this paragraph, we describe a simple algorithm for maintaining a spanner $H$ of a fully-dynamic graph $G$ with remarkably low recourse. The construction is based on the greedy spanner of \cite{althofer1993sparse}.\footnote{The low-recourse property of this algorithm has been observed for graphs that only receive edge deletions in \cite{bhattacharya_et_al}.} 

\begin{theorem}[Dynamic Spanner]
\label{thm:RecourseEdgeSparsifier}
    There exists a data structure that given a fully-dynamic graph $G = (V, E, \ll)$ with polynomially bounded edge lengths, maintains a fully-dynamic $O(\log |V|)$-spanner $H$ of $G$ such that $|E(H^{(0)})| \leq \gamma_{\text{DS}} \cdot |V(H^{(0)})|$ and $\norm{H} \leq \gamma_{\text{DS}} \cdot (\norm{D_G} + |V|)$ for $\gamma_{\text{DS}} \in \tilde{O}(1)$, where $D_G$ is a dynamic object that contains all decremental updates to $G$ (deletions and vertex splits).
    
    $H$ can be maintained in polynomial time per update.
\end{theorem}
\begin{proof}
    We assume unit lengths via maintaining separate spanners for length ranges $[2^i, 2^{i + 1})$ and returning the direct product of the appropriately scaled spanners. Furthermore, we let $|V^{(0)}| = n$ and re-start after $n$ decremental updates. 
    
    Then, there can be at most $n$ vertex splits and therefore $G$ contains at most $2n$ vertices before re-starting. We initialize $H^{(0)}$ to the greedy spanner of $G^{(0)}$, i.e. we consider the edges $(u,v) \in E(G^{(0)})$ in arbitrary order and add them to $E(H^{(0)})$ if $\dist_{H^{(0)}}(u,v) > 2 \log 2n$. This ensure that the graph $H^{(0)}$ is a $O(\log |V|)$-spanner of $G^{(0)}$, and $|E^{(0)}| \leq O(n)$ directly follows from the fact that a girth $2 \log |V| + 1$ graph contains at most $O(|V|)$ vertices. 

    Whenever $G$ is updated, we perform the corresponding operation on $H$ in case of edge deletions and vertex splits such that $H \subseteq G$. Then, we iteratively check for every edge $(u,v) \in E(G) \setminus E(H)$ if $\dist_H(u,v) > 2 \log 2n$ and insert the edges for which the condition is true.

    Clearly $H$ is an $O(\log n)$ spanner of $G$ throughout because every edge is stretched by at most $2\log 2n$. Furthermore, vertex splits and edge deletions only increase the girth of $H$, and therefore $H$ contains at most $O(n)$ edges throughout. Finally, the total recourse bound $\norm{H} \leq O(n)$ follows because an edge only leaves the graph $H$ when it is deleted. 
\end{proof}

\paragraph{Low-Recourse Pivot Hierarchies via Iterated Sparsification. } Given the low-recourse vertex and edge sparsification routines described above, we are ready to give the algorithm for \Cref{thm:low_recourse_pivots}.  

Our algorithm consists of layers $0, \ldots, \Lambda - 1$ where $\Lambda = \sqrt{\log n}$.\footnote{We remark that our full algorithm uses a much more drastic size reduction to enable bootstrapping, which results in only $O(\log^\epsilon \log n)$ layers for some small constant $\epsilon < 1$. This is necessary because the stretch is powered every time we recurse.} Each layer consists of two graphs $G_i$ and $H_i$, where $H_i \subseteq G_i$. We let $G_0 = G$. Then we obtain $H_i$ from $G_i$ and $G_i$ from $H_{i - 1}$ as follows:
\begin{enumerate}
    \item $G_i$ is the output of the vertex sparsifier from \Cref{thm:RecourseVertexSparsifier} on the graph $H_{i - 1}$ for size reduction parameter $k \defeq 2^{\sqrt{\log n}}$. For simplicity, we assume that $k$ is an integer without loss of generality.
    \item  $H_i$ is the dynamic spanner from \Cref{thm:RecourseEdgeSparsifier} on the graph $G_i$.
\end{enumerate}
Whenever the an update to $G$ occurs, we pass the updates up the hierarchy until we reach the first layer $j$ that has recieved more than $\gamma_{\textit{DS}}^{j} \gamma_{\textit{VS}}^{j} 2^{\Lambda - j}$ since it was initialized. Then, we re-initialize all the layers $j, \ldots, \Lambda$ via the steps described above. Notice in particular that the recourse of each layer is much smaller than the size reduction, leading to manageable recourse over all.

We let the sets $A_i$ be the vertex sets of the graphs $G_i$ mapped back to $G$, and we let $A_{\Lambda}$ contain an arbitrary vertex $r$ in $A_{\Lambda - 1}$ (also mapped back to $G$). Furthermore, we let the pivot map $p_i(\cdot)$ be the pivot function maintained by the data structure from \Cref{thm:RecourseVertexSparsifier} at layer $i + 1$ for $i = 0, \ldots, \Lambda - 1$, and the cluster map $C_i$ is the cluster map from vertex sparsifier data structure at layer $i + 1$. Finally, we define the cluster map $C_{\Lambda}(v) \defeq A_{\Lambda}$ and pivot map $p_{\Lambda}(v)$ for every vertex $v \in A_{\Lambda}$. 

We conclude the warm up section by proving the main theorem.

\begin{proof}[Proof of \Cref{thm:low_recourse_pivots}]
    We first bound the sizes and recourse of the sets $A_0, \ldots, A_k$. The recourse of the set $A_0$ and the recourse of $G_0$ is $\norm{G}$ by definition, and the recourse of $H_0$ is at most  $\gamma_{\text{DS}}\big(\norm{G} + |V|\big)$.

    We then observe that recourse caused by non-decremental updates on $G_i$ can be ignored because they do not show up in the recourse of the subsequent spanner $H_i$ by \Cref{thm:RecourseEdgeSparsifier}. The recourse of the decremental updates to $G_{i}$ is $\gamma_{\textit{VS}} \cdot \norm{H_{i - 1}}$. The recourse of $H_i$ is $\norm{H_i} \leq \gamma_{\textit{DS}} (\norm{D_{G_i}} + |V(G_i)|)$. Therefore, $t$ updates to $G$ cause at most $O(t \cdot \gamma_{\textit{VS}}^{i} \gamma_{\textit{DS}}^{i})$ vertices to be added to $G_i$. Since we re-start every layer $i$ after $\gamma_{\textit{DS}}^{i} \gamma_{\textit{VS}}^{i} 2^{\Lambda - i}$ updates have arrived and the total amount of updates that reach layer $i$ is $\gamma_{\textit{VS}}^{i} \gamma_{\textit{DS}}^{i} n$, the total recourse of the layer is $O(2^\Lambda \cdot n \cdot \gamma_{\textit{VS}}^{i + 1} \gamma_{\textit{DS}}^{i + 1})$ and the total recourse of all the sets is at most $O(k \cdot \gamma_{\textit{VS}}^{\Lambda + 1} \gamma_{\textit{DS}}^{\Lambda + 1} \cdot n \cdot \Lambda)$ by \Cref{thm:RecourseVertexSparsifier} and \Cref{thm:RecourseEdgeSparsifier}. 

    Then, we show that the final set $A_{\Lambda - 1}$ is of small enough size for bounding the size of the clusters $C_{\Lambda - 1}(\cdot)$. By the recourse bound above, there are at most $n^{o(1)} n/k^{\Lambda - i + 1}$ updates to a level before it gets rebuilt. Therefore, the size of level $\Lambda - 1$ is at most $n^{o(1)}$ at all times. 
    The bound on the other cluster sizes directly follows from \Cref{thm:RecourseVertexSparsifier} and the description of our algorithm. 

    It remains to show that the distance estimates $\tilde{\dist}$ are good. To do so, we first notice that for every $u, v \in A_i$ we have $\dist_G(u,v) \leq \dist_{G_i}(u,v) \leq n^{o(1)} \dist_G(u,v)$ since each layer only loses $\tilde{O}(1)$ in distance approximation by \Cref{thm:RecourseVertexSparsifier} and \Cref{thm:RecourseEdgeSparsifier} and there are $\sqrt{\log n}$ levels. Then, the distance approximation follows by induction since the length of the detour to the pivots can lose at most a poly-logarithmic factor per level because the pivot is closer in the graph $G_i$ and $\dist_G(u,v) \leq \dist_{G_i}(u,v) \leq n^{o(1)} \cdot \dist_G(u,v)$. 
\end{proof}

\section{The Algorithm} \label{sec:algorithm}

We solve the APSP problem on a large graph $G$ using a pivot hierarchy obtained by iterated vertex and edge sparsification, 
just like in the warm up in \Cref{sec:warm_up}.
Again, we will select a set of pivots for the first level of a pivot hierarchy, and perform edge and vertex sparsification to represent the distances between the pivots using a smaller dynamic graph, and then recurse on this graph.
The main difference to the previous section is that (a) we want to bootstrap dynamic APSP along the way and (b) we want to develop and use a computationally efficient spanner.
This spanner creates a few technical complications.
In particular, to make the spanner efficient, we need all our graphs to have somewhat bounded maximum degree.

\paragraph{Vertex sparsification. }  First, our algorithm reduces the APSP problem to a graph $\tilde{G}$ with significantly fewer vertices using the KMG-vertex sparsifier (See \Cref{sec:vertex_sparsification}). When compared to the low-recourse version presented in \Cref{sec:warm_up}, our vertex sparsifier has an additional routine called $\textsc{ReduceDegree}(\cdot)$ which forces some of the vertices in the vertex sparsifier to be split and therefore have smaller degree. We further elaborate this point in the paragraph below. 

The vertex sparsification algorithm is presented in \Cref{sec:vertex_sparsification}.

\paragraph{Edge sparsification.} Although $\tilde{G}$ is supported on a much smaller vertex set, it may still contain most of the edges from $G$. To achieve a proper size reduction, we additionally employ edge sparsification on top of the vertex sparsifier $\tilde{G}$.

There are two key differences between the warm up and the efficient spanner algorithm. 

\begin{itemize}
    \item To quickly check for detours in the spanner, the efficient algorithm recursively uses APSP datastructures on small instances (See \Cref{sec:edge_sparsification} and \Cref{sec:spanner}). 
    \item To efficiently maintain the dynamic spanners, it is crucial that the degree of the graph remains bounded. Although we can assume that the degree of the initial graph $G$ is bounded by a constant, repeatedly sparsifying the graph could seriously increase its maximum degree. Fortunately, our dynamic spanner guarantees that the average degree after edge sparsification is low. However, the same cannot be achieved for the maximum degree, which becomes apparent when trying to sparsify a star graph. 
    
    To remedy this issue, we use that the underlying graph has bounded degree and call the $\textsc{ReduceDegree}(\cdot)$ routine of the KMG-vertex sparsifier to reduce the degree of high-degree vertices in the spanner. This operation can lead to changes in $\tilde{G}$ and its spanner, but we show that repeatedly splitting vertices with high degree in the spanner quickly converges because the progress of splits is larger than the recourse they cause to the spanner.  
\end{itemize}

The edge sparsification algorithm is presented in \Cref{sec:edge_sparsification} and \Cref{sec:spanner}.

\subsection{Pivots and Vertex Sparsification}
\label{sec:vertex_sparsification}

We use the KMG-vertex sparsifier algorithm \cite{kyng2023dynamic}, which is based on dynamic core graphs, which in turn are based on static low-stretch spanning trees. 

We first introduce certified pivots. These are mapping vertices that get sparsified away to vertices in $\tilde{G}$, and certify the distance for vertices that are closer to each other than to their respective pivots. 

\begin{definition}[Certified Pivot]
    \label{def:certified_pivot}
    Consider a graph $G=(V,E)$ with edge lengths $\ll$ and a set $A \subseteq V$.
    We say that a function $p : V \rightarrow A$ is a pivot function for $A$, if for all $v \in V$, we have that $p(v)$ is the nearest vertex in $A$ (ties broken arbitrarily but consistently).

    Additionally, we say that $p$ is a certified pivot function if
    for all vertices $v \in V$, for all $u \in B_G(v, p(v)) \cup \{p(v)\}$
    we have computed the exact $(u,v)$ shortest paths and distances.
\end{definition}

To motivate \Cref{def:certified_pivot}, we show that pivots can be used to approximate the distance between sparsified vertices whenever their exact distance isn't stored already.  

\begin{claim}\label{thm:pivotdist}
Consider a graph $G = (V, E, l)$ with pivot set $A$ and corresponding pivot function $p$. Then for any $u, v \in V$, if $v \not \in B_G(u, p(u))$, we have $\dist_G(p(u), p(v)) \leq 4 \cdot \dist_G(u, v)$. 
\end{claim}
\begin{proof}
By $v \not \in B_G(u, p(u))$ we directly obtain $\dist_G(u, p(u) \leq \dist_G(u, v)$ and since $p(v)$ is (one of) the neareast vertices to $v$ in $A$, we have $\dist_G(v, p(v)) \leq \dist_G(v, p(u)) \leq \dist_G(v, u) + \dist_G(u, p(u)) \leq 2 \dist_G(u, v)$ as well.  
By using the triangle inequality we get $\dist_G(p(u), p(v)) \leq \dist_G(p(u), u) + \dist_G(u, v) + \dist_G(v, p(v)) \leq 4 \dist_G(u, v)$.
\end{proof}

Next, we state a version of the KMG-vertex sparsifier \cite{kyng2023dynamic}, which is based on dynamic core graphs. We refer the reader to \Cref{apx:vertex_sparsifier} for a more detailed discussion of their vertex sparsifier. 

The following theorem should be thought of as an efficient analog to  \Cref{thm:RecourseVertexSparsifier} with the extra operation $\textsc{ReduceDegree}(\cdot)$. As mentioned at the start of this section, this routine comes into play when the spanner produces a high-degree vertex down the line. 

\begin{restatable}[KMG-vertex sparsifier, See Theorem 3.1 in \cite{kyng2023dynamic} and \Cref{apx:vertex_sparsifier}]{theorem}{vertexSparsifier}\label{thm:VertexSparsifier}
Consider a size reduction parameter $k > 1$ and an edge-dynamic graph $G$, with $\Delta_{\max}(G) \leq \Delta$ and lengths in $[1, L]$. Then, there is a deterministic algorithm that explicitly maintains, 
\begin{enumerate}
    \item a monotonically growing pivot set $A \subseteq V(G)$ and certified pivot function $p : A \rightarrow V$, such that  $|A| \leq n/k + 2\norm{G}$.
    \item \label{prop:workhorseDistancePreserve} a fully-dynamic graph $\Tilde{G}$, such that $\Tilde{G}$ is a $\gamma_{\textit{VS}}$-approximate vertex sparsifier of $G$ with respect to $A$. We have that $\Tilde{G}^{(0)}$ has at most $m \cdot \gamma_{\textit{VS}}$ edges and $\gamma_{\textit{VS}} \cdot m /k$ vertices, where $\gamma_{\textit{VS}} = \Otil(1)$. The following properties hold:
    \begin{enumerate}
        \item \label{prop:maxDegreeVS} $\Delta_{\max}(\Tilde{G}) \leq \Delta \cdot  \gamma_{\textit{VS}} \cdot k$, and $\Tilde{G}$ has lengths $\ll_{\Tilde{G}}$ in $[1, nL]$, and 
        \item given any edge $e = (u,v) \in \Tilde{G}$, the algorithm can return a $uv$-path $P$ in $G$ with $l_G(P) \leq l_{\Tilde{G}}(e)$ in time $O(|P|)$. \label{prop:mapBackToPath}
        \item a procedure $\textsc{ReduceDegree}(E', z)$, where $E' \subseteq \{e \in E(\Tilde{G}): v \in e\}$ for some $v \in V(\Tilde{G})$, and $z$ is a positive integer. It updates the graph $\Tilde{G}$, performing at most $\gamma_{\textit{VS}} \cdot |E'|/z$ vertex splits, such that afterwards any vertex $v' \in V(\Tilde{G})$ with $v \in \master(v')$ is adjacent to at most $z \cdot \Delta$ edges in $E'$.  This procedure runs in time at most $\gamma_{\textit{VS}} \cdot k \cdot |E'|$.
        \item\label{prop:RD} If $R_V$ is the dynamic object containing the total number of vertex splits performed by \textsc{ReduceDegree}, then $\norm{\Tilde{G}} \leq \gamma_{\textit{VS}} \cdot \norm{G} + \norm{R_V}$.
    \end{enumerate}
\end{enumerate}
The algorithm runs in time $m \cdot \gamma_{\textit{VS}} \Delta k^4 + \norm{G} \cdot \gamma_{\textit{VS}}k^4 \Delta$.
\end{restatable}

\subsection{Edge Sparsification}
\label{sec:edge_sparsification}

We state our main dynamic spanner theorem whose proof is deferred to \Cref{sec:spanner}. The theorem assumes access to a $\gamma_{apxAPSP}$-approximate, $\Otil(1)$-query APSP data structure (See \Cref{def:apsp}.)

\begin{theorem}[Dynamic Spanner]
\label{thm:EdgeSparsifier}
    There exists a data structure that given a fully-dynamic graph $G$ with unit lengths and $|V(G^{(0)})| = n$, $\norm{G} \leq n$ and a parameter $1 \leq K \leq O(\log^{1/3} n )$, maintains a fully-dynamic $\gamma_l^{O(K)}$-spanner $S$ of $G$ such that $|E(S)| \leq \gamma_{\textit{ES}} \cdot n$ and $\norm{S} \leq \gamma_{\textit{ES}} \cdot \norm{G}$, where $\gamma_l = \Otil(\gamma_{apxAPSP})$ and $\gamma_{\textit{ES}} = \Otil(n^{1/K})$. It runs in time $n\Delta_{\max}(G^{(0)}) \gamma_{\textit{ES}}\gamma_{l}^{O(K^2)} + APSP(\gamma_{\textit{ES}} n, 3, \gamma_{\textit{ES}} n)$.
\end{theorem}

\begin{remark}
    \label{rem:unit_length}
    We can directly extend \Cref{thm:EdgeSparsifier} to graphs with edge lengths in $[1, L]$ by bucketing edges in the intervals $[2^i, 2^{i+1})$ for $i = 1, \ldots, O(\log L)$.

    Note also that as at all times $H$ is a subgraph of $G$, it does not undergo more vertex splits and isoalted vertex insertions than $G$, i.e., all extra updates to maintain $H$ are edge insertions/deletions.
\end{remark}

\subsection{Bootstrapping via Edge and Vertex Sparsification}
We are given an edge-dynamic graph $G$ on initially $n$ vertices and $m$ edges, that at all times has maximum degree at most $\Delta$. Applying \Cref{thm:VertexSparsifier} to $G$ reduces the problem to finding short paths in $\Tilde{G}$, a graph of $\Otil(m/k)$ vertices. Although the number of vertices significantly decreases, the graph still contains up to $\Otil(m)$ edges, preventing efficient recursion. To address this, we apply \Cref{thm:EdgeSparsifier} on $\Tilde{G}$ to produce a graph $H$ on $\gamma \cdot m / k$ vertices \emph{and} edges, where $\gamma$ crucially depends only on the parameter $K$ of \Cref{thm:EdgeSparsifier}.

Note that the APSP data structure that \Cref{thm:EdgeSparsifier} requires also only needs to run on a graph of at most $\gamma \cdot m/k$ vertices and edges, so by choosing a large enough size reduction factor $k$, the problem size reduces by a factor of $\gamma/k$, while the approximation factor increases from $\gamma_{apxAPSP<m}$ to $\Otil(\gamma_{apxAPSP<m})^{O(K)}$.\footnote{We only have a valid size reduction if also $\gamma \cdot m/k < n$. We will see that this is indeed the case as we can assume that initially $m = O(n)$.}  Balancing both parameters $k$ and $K$ allows us to bootstrap an APSP data structure with the guarantees from \Cref{thm:main}.

While $H$ is sparse, it might still have high maximum degree, preventing efficient recursion. The sparsity still implies, however, that there cannot be too many high degree vertices. Hence, we carefully perform a few vertex splits in $\Tilde{G}$ by leveraging the $\textsc{ReduceDegree}(\cdot)$ functionality from \Cref{thm:VertexSparsifier} to enforce that $H$ also has small maximum degree.

\begin{algorithm}
Maintain $\tilde{G}$ from \Cref{thm:VertexSparsifier}. \\
Let $\tilde{H}$ be the graph maintained by applying \Cref{thm:EdgeSparsifier} to graph $\tilde{G}$ with parameter $K$. \\
\tcc{ Here $\gamma_{degConstr}$ is a constant fixed later. }
\While{$\exists v \in V(H)$ with $\deg_{H}(v) > 8 \gamma_{degConstr} \cdot \Delta$}{\label{algo:initDynamicAPSP:while}
    $\textsc{ReduceDegree}(E_{H}(v), \gamma_{degConstr})$.
}
Initialize and maintain a recursive dynamic APSP data structure on $H$.
\caption{$\textsc{initDynamicAPSP}()$}
\label{algo:initDynamicAPSP}
\end{algorithm}

The recursive APSP instance that is running on $H$ expects an edge-dynamic graph, while the spanner from \Cref{thm:EdgeSparsifier} is fully-dynamic. To circumvent this issue, we simulate the vertex splits by suitable edge insertions/removals and vertex insertions. As the maximum degree of $H$ gets controlled during the while loops, this simulation can lead to at most an extra factor of $O(\gamma_{degConstr} \cdot \Delta)$ in the number of updates.

\begin{algorithm}
Update $\tilde{G}$ and $H$ accordingly.\\
\While{$\exists v \in V(H)$ with $\deg_{H}(v) > 8 \gamma_{degConstr} \cdot \Delta$}{ \label{algo:UpdateAPSP:while}
    $\textsc{ReduceDegree}(E_{H}(v), \gamma_{degConstr})$. \\
    Update $H$ accordingly. 
}
Forward all updates made to $H$ to the recursive dynamic APSP datastructure that runs on $H$ by simulating the vertex splits accordingly
\caption{$\textsc{MaintainDynamicAPSP}(G, t)$}
\label{algo:UpdateAPSP}
\end{algorithm}
When queried for the distance between two vertices $u$ and $v$, our algorithm checks if the certified pivot function stores the distance, and if this is not the case it recursively asks the smaller instance $\textsc{DynamicAPSP}_{H}$ for the distance between $p(u)$ and $p(v)$ in $H$. 
\begin{algorithm}
\If{$v \in B_{G}\big(u, A\big)$}{
    \Return $\dist(u, v)$
}
\Else{
    \Return $\dist(u, p(u)) + \textsc{DynamicAPSP}_{H}.\textsc{Dist}(p(u),p(v)) + \dist(p(v), v)$ 
}
\caption{$\textsc{Dist}(u,v)$}
\label{alg:query}
\end{algorithm}

\subsection{Analysis}

We first show that the while loops in \Cref{algo:initDynamicAPSP} and \Cref{algo:UpdateAPSP} terminate sufficiently fast. This is the most crucial claim for the analysis. 

\begin{claim}
    If $R_V$ is the dynamic object counting the total number of vertex splits performed by $\textsc{ReduceDegree}()$ in \Cref{algo:initDynamicAPSP} and \Cref{algo:UpdateAPSP}, then $\norm{R_{V}} \leq 2\gamma_{\textit{VS}} (m / k + \norm{G})$. In particular, the while loops always terminate and immediately thereafter $H$ has maximum degree at most $8 \gamma_{degConstr} \cdot \Delta$. The total runtime of all calls to $\textsc{ReduceDegree}()$ is $6\gamma_{\textit{VS}}^2 \gamma_{\textit{ES}} (m + k \norm{G})$.
\end{claim}
\begin{proof}
For tracking progress achieved by vertex splits, we introduce the following potential function
\[\Phi(H) \defeq \sum_{v \in V(H)} \max\{\deg_{H}(v) - \Delta \cdot \gamma_{degConstr}, 0\}.
\]

Let $t$ be the time before we enter the while loop at \Cref{algo:initDynamicAPSP:while} of \Cref{algo:initDynamicAPSP}. Then by \Cref{thm:VertexSparsifier} we have that $|V(\Tilde{G}^{(t)})| \leq \gamma_{\textit{VS}} m /k$ and thus by \Cref{thm:EdgeSparsifier}, we have that $|E(H^{(t)})| \leq \gamma_{\textit{VS}} \gamma_{\textit{ES}} m / k$. 

By the Handshake-Lemma, $\Phi(H^{(t)}) \leq 2 |E(H^{(t)})| \leq 2 \cdot \gamma_{\textit{ES}} \cdot \gamma_{\textit{VS}} \cdot m/k$. Note that whenever an update is made to $H$, the potential can increase by at most $2$. This is as edge deletions and vertex splits can only decrease it, and an edge insertion can increase it by at most $2$. 

Let $R_E$ be the dynamic object counting the number of edges that were passed to $\textsc{ReduceDegree}$, and note that if $\textsc{ReduceDegree}(E', \gamma_{degConstr})$ is called, we have $|E'| > 8 \gamma_{degConstr} \cdot \Delta$ by the while loop condition. 

By the properties of $\textsc{ReduceDegree}()$ in \Cref{thm:VertexSparsifier}, we have that $\norm{R_V} \leq \gamma_{\textit{VS}} / \gamma_{degConstr} \cdot \norm{R_E}$, and that none of the vertices $v'$ that result from the splits to $v$ will be adjacent to more than $\gamma_{degConstr} \cdot \Delta$ of the edges in $E'$. So if we just forwarded these splits to $H$, we would have $\deg_{H}(v') \leq \gamma_{degConstr} \cdot \Delta$, and the potential decreases by at least $7/8 \cdot |E'|$. 

However, forwarding these vertex splits to $H$, \Cref{thm:EdgeSparsifier} will cause further updates to $H$ to maintain the spanner. Fortunately by \Cref{prop:RD} of \Cref{thm:VertexSparsifier} we know that $\norm{H} \leq \gamma_{\textit{ES}} \cdot \norm{\Tilde{G}} \leq \gamma_{\textit{ES}} \cdot \gamma_{\textit{VS}} \cdot \norm{G} + \gamma_{\textit{ES}} \cdot \norm{R_V}$. Let $\gamma = \gamma_{\textit{ES}} \cdot \gamma_{\textit{VS}}$. Putting everything together we show that the potential decreases as follows. 
\begin{align*}
& \Phi(H) \leq 2 \gamma m/k - \frac{7}{8}\norm{R_E} + 2\norm{H} \leq 2 \gamma m/k - \frac{7}{8} \norm{R_E} + 2\gamma\norm{G} + 2\gamma_{\textit{ES}}\norm{R_V} \\
& \leq 2 \gamma m/k - \frac{7}{8} \norm{R_E} + 2\gamma \norm{G} + \frac{2 \gamma_{\textit{ES}} \gamma_{\textit{VS}}}{\gamma_{degConstr}} \norm{R_E}.
\end{align*}
Thus, choosing $\gamma_{degConstr} = \lceil 4 \gamma_{\textit{VS}} \gamma_{\textit{ES}} \rceil$, we get $\Phi(H) \leq 2 \gamma m/k - \frac{3}{8} \norm{R_E} + 2\gamma \norm{G}$.
As $0 \leq \Phi(H)$, this directly implies that $\norm{R_E} \leq 6 \gamma_{\textit{ES}} \gamma_{\textit{VS}}(m/k + \norm{G})$. Remembering that $\norm{R_V} \leq \frac{\gamma_{\textit{VS}}}{\gamma_{degConstr}}\norm{R_E}$ then yields the bound on $\norm{{R}_V}$. 

Since the total runtime is given by $\gamma_{\textit{VS}} \cdot k \cdot \norm{R_E}$, this bound follows directly. 
\end{proof}

We then show a bound on the total runtime of our dynamic APSP algorithm in terms of costs for smaller instances, setting us up for our final recursion.

To do so, assume the recursive APSP instance that is running on $H$ and used in \Cref{thm:EdgeSparsifier} is a $\gamma_{apxAPSP<m}$-approximate, $\Otil(1)$-query APSP. Let $\gamma_l = \Otil(\gamma_{apxAPSP<m})$ be the stretch parameter from \Cref{thm:EdgeSparsifier}.

\begin{theorem}
    \label{thm:recursive_runtime}
    The total combined runtime of \Cref{algo:initDynamicAPSP} and \Cref{algo:UpdateAPSP} over a sequence of up to $m/k$ updates is $m \cdot \gamma_{RC} \cdot \gamma_{l}^{O(K^2)} \Delta k^4 + APSP(\gamma_{RC} m/k, \gamma_{RC} \Delta, \gamma_{RC} \Delta  m/k)$ for some $\gamma_{RC} = \Otil(m^{2/K})$.
\end{theorem}
\begin{proof}
We first show that for at most $m/k$ updates, initializing and maintaining $H$ is efficient. Note that by the previous claim, we can bound the total costs incurred in the while loops by $6\gamma_{\textit{VS}}^2 \gamma_{\textit{ES}} (m + k \norm{G}) \leq 12 \gamma_{\textit{VS}}^2 \gamma_{\textit{ES}} m$. The total costs of initializing and maintaining $\Tilde{G}$ are $2m \cdot \gamma_{\textit{VS}} \Delta k^4$ by \Cref{thm:VertexSparsifier}. Finally, by the previous lemma $\lVert{\Tilde{G}}\rVert \leq \gamma_{\textit{VS}} \norm{G} + \norm{R_V} \leq 3\gamma_{\textit{VS}} \norm{G} + 2 \gamma_{\textit{VS}} m/k \leq 4 \gamma_{\textit{VS}} m/k$, so that we can further bound the runtime of \Cref{thm:EdgeSparsifier} by $\gamma_{\textit{VS}} \gamma_{\textit{ES}} n \Delta \gamma_l^{O(K^2)} + APSP(\gamma_{\textit{ES}} \cdot n, 3, \gamma_{\textit{ES}} \cdot n)$. 

It remains to consider the APSP data structure that we have to run on $H$. First note that by the gurantees of \Cref{thm:VertexSparsifier} and \Cref{thm:EdgeSparsifier}, initially $|E(H)| \leq \gamma_{\textit{VS}} \gamma_{\textit{ES}} m/k$. Secondly, note that as $\lVert{\Tilde{G}}\rVert \leq 4 \gamma_{\textit{VS}} m/k$, we in particular have $\norm{H} \leq 4 \gamma_{\textit{ES}} \gamma_{\textit{VS}} m/k$. However, as $H$ is vertex dynamic, we also have to simulate all the vertex splits occuring in $H$, which can lead to up to $8 \gamma_{degConstr} \Delta$ extra updates per split. Hence the total amount of updates that the APSP data structure receives are up to $O(\gamma_{\textit{VS}}^2 \gamma_{\textit{ES}}^2 m/k)$. By the while loop conditions, we do know that the maximum degree of the graph is always bounded by $8 \gamma_{degConstr} \Delta$. Therefore, the total runtime bound follows. 
\end{proof}

We show that the stretch does not increase too much as the size of the graph increases.

\begin{claim}\label{lma:DynamicAPSPValidity}
\Cref{algo:initDynamicAPSP} provides a worst case stretch of $\Otil(\gamma_{apxAPSP<m})^{O(K)}$.
\end{claim}
\begin{proof}
    If $v \in B_G(u, A)$ there is nothing to show as then by \Cref{thm:VertexSparsifier} we have already pre-computed the exact distance between $u$ and $v$. 
    
    But if $v \not \in B_G(u, A)$, then by \Cref{thm:pivotdist}, we have that $\dist_G(p(u), p(v)) \leq 4 \cdot \dist_G(u, v)$. As $p(u), p(v) \in A$, we have $\dist_{H}(p(u), p(v)) \leq \gamma_l^{O(K)} \cdot \dist_{\Tilde{G}}(p(u), p(v)) \leq \gamma_l^{O(K)} \cdot \gamma_{\textit{VS}} \cdot \dist_{G}(p(u), p(v)) = \Otil(1) \cdot (\Otil(1) \cdot \gamma_{apxAPSP})^{O(K)} \cdot \dist_G(p(u), p(v)) = \Otil(\gamma_{apxAPSP})^{O(K)} \cdot \dist_G(p(u), p(v))$. 
\end{proof}

\begin{remark}
    We can directly extend this to receive a procedure $\textsc{Path}(u,v)$ that provides a $u$-$v$ path $P$ with $|P| \leq \Otil(\gamma_{apxAPSP<m})^{O(K)} \cdot \dist_G(u, v)$ and runs in time $\Otil(|P|)$. This is because \Cref{thm:VertexSparsifier} also explicitly stores exact shortest paths between vertices $u$ and their pivots $p(u)$. 
\end{remark}

We conclude this section by proving our main result, which we restate for convenience.

\mainTheorem*

\begin{proof}
First, by replacing vertices with Binary-Search-Trees as mentioned at the end of \Cref{fact:BSTred}, we can assume that $\Delta_{\max}(G) \leq 3 = \Delta_0$ and $|E(G)| = O(|V(G)|)$.

We recursively apply \Cref{algo:initDynamicAPSP} and \Cref{algo:UpdateAPSP}. We re-start both algorithms whenever the input graph has received more than $m/k$ updates. We fix the parameters as $k = m^{1 / (\log \log m)^{1/4}}$ and $K = (\log \log m)^{3/4}$. If $\Lambda := \log_k(m) = (\log \log m)^{1/4}$, then, as $\gamma_{RC} \in \Otil(m^{2/K})$, we in particular have $\gamma_{RC}^{\Lambda} = o(k)$. By \Cref{thm:recursive_runtime}, at each level of the recursion, the problem size decreases by a factor of $\gamma_{RC} /k$. Once a graph has size $\leq k$, we are done as then \Cref{thm:VertexSparsifier} directly computes all exact shortest paths. Hence, the maximal recursion depth of the algorithm is $\Lambda$. Repeatedly using \Cref{lma:DynamicAPSPValidity} then yields a worst case stretch of $\gamma_{apxAPSP} = \Otil(1)^{O(K)^{\Lambda}} = m^{1/\log(m)^{1-o(1)}} = m^{o(1)}$. In particular, we also have $\gamma_l = \Otil(\gamma_{apxAPSP})^{O(K^2)} = m^{o(1)}$.

It remains to analyse the runtime. To do so, let us denote by $\Delta_i$ the maximum degree at recursion depth $i$, and note that $\Delta_i \leq \gamma^i \Delta_0$, again by \Cref{thm:recursive_runtime}. Applying the same theorem we obtain
\begin{align*}
        & APSP(m, \Delta_0, m) = k APSP(m, \Delta_0, m/k) \\
        & \leq k (m \cdot \gamma_{RC} \cdot \gamma_l^{O(K^2)}\Delta_0 k^4 + APSP(\gamma_{RC} m/k, \gamma_{RC} \Delta_0, \gamma_{RC} \Delta_0 m/k)) \\
        & \overset{(i)}{=} m \cdot \gamma_{RC} \cdot \gamma_l^{O(K^2)}\Delta_0 k^5 + k^2 \Delta_0 APSP(\gamma_{RC} m/k, \Delta_1, \gamma_{RC} m /k^2) \\ & \leq m \cdot \gamma_{RC} \cdot \gamma_l^{O(K^2)}k^5 \cdot \big(\Delta_0 + \Delta_0 \Delta_1\big) + k^3 \Delta_0 \Delta_1 APSP(\gamma_{RC}^2 m/k^2, \gamma_{RC} \Delta_1, \gamma_{RC}^2 m/k^2) \\
        & \leq \cdots \leq \sum_{i=0}^{\Lambda} m \cdot \gamma_{RC} \cdot \gamma_l^{O(K^2)} k^5 \prod_{j \leq i} \Delta_j \leq \sum_{i=0}^{\Lambda}  m \cdot \gamma_l^{O(K^2)} k^5 \prod_{j \leq i} \gamma_{RC}^{j+1} \Delta_0  \\
        & \overset{(ii)}{\leq} 3 \Lambda m^{1+o(1)}k^5 \cdot \gamma_{RC}^{O(\Lambda^2)} = m^{1+o(1)+O(1/(\log \log m)^{1/4})} = \cdot m^{1 + o(1)},
    \end{align*}
    where $(i)$ follows as we re-start the algorithm after $\gamma_{RC} m / k^2$ updates are processed and in $(ii)$ we used that $\gamma_l^{O(K^2)} = m^{o(1)}$ as was proven above. This concludes the runtime analysis. 
    
    The query time is $O(1)$ per level, so the total query runtime is $O(\Lambda) = O(\log \log m)$.
\end{proof}

\section{Dynamic Spanner via APSP}
\label{sec:spanner}
In this section we prove \Cref{thm:EdgeSparsifier}, assuming access to a $\gamma_{apxAPSP}$-approximate, $\Otil(1)$-query APSP data structure (See \Cref{def:apsp}).
This theorem gives a computationally efficient version of the low-recourse spanner of \Cref{thm:RecourseEdgeSparsifier}.

\paragraph{Using APSP queries in a spanner.}
When given access to an APSP data structure, a simple strategy to improve the runtime of the spanner from \Cref{thm:RecourseEdgeSparsifier} is to replace the distance queries with calls to the data structure. Since we re-insert all the edges in the graph after every edge deletion/vertex split, processing such an update unfortunately still costs time at least linear in the total number of edges. 

\paragraph{Controlling congestion.} To avoid having to re-insert all the edges, we can try remembering the paths that witness short distances for the edges in the graph $G$ but not its spanner $H$. Then, if an edge is deleted from $H$ because of a deletion to $G$, we only have to re-insert all edges that have a broken path. We call such a map from edges to paths an edge embedding. 

Ideally, we want this edge embedding to ensure that 
every edge and every vertex of $H$ only has few edges of $G$ embedding through it.
We call this low \emph{congestion}.
When the edge embedding has low congestion, we can hope to efficiently repair it as $G$ and $H$ change.
It turns out that we can fairly easily control the number of edges of $G$ that embed into each edge of $H$. 
Unfortunately, $H$ may contain high-degree vertices, and naively embedding edges causes too high congestion of these vertices.
We therefore build $K + 1$ separate spanners $H_0, \ldots, H_K$, where the spanner at level $0$ does not constrain the edge congestion very much, but limits the degree. As $i$ increases, we allow the degree to become larger while trading it off with a stricter bound on the edge congestion. 

\paragraph{Controlling recourse with \emph{patching}.}
Once we have a spanner with bounded congestion, the number of edges to check after a broken embedding is bounded.
However, we also want to ensure that our spanner does not undergo too many insertions after each edge deletion.
To establish this, we use a fundamentally different strategy than \Cref{thm:RecourseEdgeSparsifier}.
We will explain how to deal with edge deletions, but this strategy
seamlessly extends to dealing with
vertex splits as well.

We will describe a strategy that allows us to repair the spanner $H$ after $G$ and $H$ undergo deletion of a (potentially large) set of edges $E'$.
Whenever a set of edges $E'$ is deleted from the spanner $H$, we project all the edges of $G$ with broken embedding paths onto the set of vertices incident to deleted edges by mapping each endpoint of the edge to the first vertex on the path that is incident to a deletion. We call the resulting `projected' graph $J'$ and compute a spanner $J$ of $J'$. 
Whenever an edge is chosen to be in the spanner $J$, we use the corresponding edge of $G$ to repair broken embeddings while carefully ensuring that $J$ is picked such that the congestion of the resulting edge embedding is not too high. We say that these edges patch the spanner after being projected back. See \Cref{fig:patch} for a small illustrative example. 

\begin{figure}[H]
    \centering
    \includegraphics[width=0.65\linewidth]{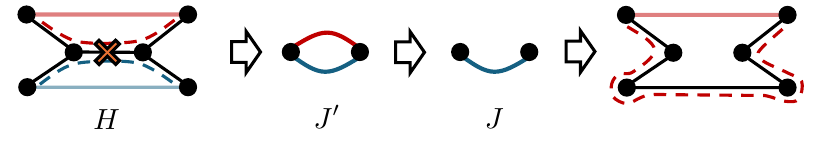}
    \caption{The black graph $H$ is a spanner of the graph $G$ that also contains the red and blue edge. After the central edge is removed, the red and blue edge are projected onto the graph $J'$ supported on the two vertices incident to the deletion. Then this graph is sparsified, and the spanner $H$ is patched with the remaining blue edge. The embedding path of the red edge is repaired using the previous blue path and the new edge. When sparsifying $J'$ we ensure that the extra vertex congestion caused by these repaired paths in $H$ only increases by a small sub-polynomial factor.}
    \label{fig:patch}
\end{figure}

Crucially, the number of edges in $J$ scales in the number of vertices in $J'$,
which in turn is bounded by $2 |E'|$.
This is generally much smaller than the number of edges in $J'$.
This ensures the recourse of repairing the spanner $H$ after the deletion of $E'$ is small.
We refer to this strategy for repairing $H$ as \emph{patching} the spanner.

\paragraph{Batching and patching.}

One obstacle remains:
each time we patch our spanner, the congestion of our edge embedding increases slightly, and the distance approximation of the spanner worsens by a polylogarithmic factor.
This means we have to ensure the sequential depth of our patching steps is very small.
Fortunately, we can do this with a standard batching scheme.
To this avail, consider a balanced tree with branching factor $n^{o(1)}$ that has the ordered sequence of updates as leaves. We then obtain a patched spanner by following the path to the newest update in the tree, and whenever we are at a node, we collect all updates to the left into a batch. This procedure divides the updates into batches such that the sequential depth is low, and large batches change their contents infrequently. The total time overhead is only a $n^{o(1)}$ factor over  the time spent to apply a single large batch.

\subsection{Patching Spanner}
\label{sec:comp_spanner}
\paragraph{Edge Embeddings. } As outlined above, we introduce edge embeddings to achieve an efficient spanner. These map all edges in an original graph $G$ to paths in its spanner $H \subseteq G$.

\begin{definition}[Edge Embedding]
Given two graphs $H, G$ and an injective map $\Pi: V(G) \rightarrow V(H)$, we call a map $\Pi_{G \rightarrow H}$ an edge embedding of $G$ into $H$ if it maps each edge $e = (u,v) \in E(G)$ into a $xy-$path $\Pi_{G \rightarrow H}(e)$ in $H$, where $x = \Pi(u), y = \Pi(v)$. 

For an embedding $\Pi_{G \rightarrow H}$, we then define its length as $\length(\Pi_{G \rightarrow H}) \defeq \max_{e \in E(G)} |\Pi_{G \rightarrow H}(e)|$. If $G'$ is a subgraph of $G$, we let $\Pi_{G' \rightarrow H}$ denote the embedding where we simply inherit the embedding paths for the subset of edges in $G'$.
\end{definition}
Now, if $H\subseteq G$ and $\Pi_{G \rightarrow H}$ is an embedding with $\length(\Pi_{G \rightarrow H}) \leq \alpha$, then this immediately acts as a certificate that $H$ is an $\alpha$-spanner of $G$. The strategy of our dynamic algorithm is be to explicitly maintain such an embedding. The advantage is that whenever an update is made to $G$ and forwarded to $H$, we only need to deal with those edges whose embedding paths are affected by the updates forwarded to $H$. By enforcing low \emph{vertex congestion} of our embeddings, we can keep the number of such edges small enough.

\begin{definition}[Edge and Vertex Congestion]
    For an edge $e^{H} \in E(H)$, we let $\econg(\Pi_{G \rightarrow H}, e^{H})$ denote the total amount of times the edge $e^H$ appears on paths in $\Pi_{G \rightarrow H}$\footnote{Since the embedding paths are not guaranteed to be simple, a path can contribute more than once to edge and vertex congestion.} and $\econg(\Pi_{G \rightarrow H}) \defeq \max_{e^H \in E(H)} \econg(\Pi_{G \rightarrow H}, e^H)$. Similarly, we let $\vcong(\Pi_{G \rightarrow H}, v^{H})$ denote the total amount of times the vertex $v^H$ appears on paths in $\Pi_{G \rightarrow H}$ and $\vcong(\Pi_{G \rightarrow H}) \defeq \max_{v^H \in V(H)} \vcong(\Pi_{G \rightarrow H}, v^H)$. 
\end{definition}

In this section, we present an algorithm that can also be used to initialize a spanner $H$ and embedding from $G$ into $H$ with $\vcong(\Pi_{G \rightarrow H}) \approx \Delta_{max}(G)$. This means that after the first update is made to $G$, we only need to deal with $\approx \Delta_{max}(G)$ edges.

Our algorithm can, however, do more than that. In \Cref{sec:single_batch}, we show how to repair our embeddings by finding short paths between pairs of so called \emph{touched vertices}. These paths will be of the form $[\Pi_{J' \rightarrow H} \circ \Pi_{J \rightarrow J'}](e)$, where $J'$ is an auxiliary graph that contains touched vertices as vertex set, and $J$ is an edge sparsifier of $J'$. Generally the vertex set of $J'$ will be significantly smaller than $|V(H)|$. Our algorithm will also be to find such a spanner $J$ while simultaneously making sure that the composed paths are still short and have relatively small vertex congestion.

\paragraph{The Static Patching Spanner Theorem. } We state the main theorem of this section. On an initial read, the reader might want to assume that $H = J'$ and $\Pi_{H \rightarrow J'}$ is the identity map, in which case the algorithm $\textsc{Sparsify}$ is simply computing a low vertex congestion spanner $J$ of $J'$ (See \Cref{lma:initspanner}). Later, we will need that $H$ can be a much larger graph than $J$.

\begin{theorem}[Patching Spanner]\label{lma:CompSpanner} Given unweighted, undirected graphs $H$ and $J'$, such that $V(J') \subseteq V\left(H\right)$, and an embedding $\Pi_{J' \rightarrow H}$ from $J'$ into $H$ there is a deterministic algorithm $\operatorname{SPARSIFY}\left(H, J', \Pi_{J' \rightarrow H}\right)$, that returns a sparsifier $J \subseteq J'$ with an embedding $\Pi_{J' \rightarrow J}$ from $J'$ to $J$ such that:
	\begin{enumerate}\item\label{lma:CompSpanner:vcong} $\vcong\left(\Pi_{J' \rightarrow H} \circ \Pi_{J' \rightarrow J}\right) \leq \Otil(\gamma_{apxAPSP}) \cdot \length\left(\Pi_{J' \rightarrow H}\right) \cdot\left(\vcong\left(\Pi_{J' \rightarrow H}\right)+\Delta_{\max }(J')\right)$, 
	\item\label{lma:CompSpanner:length}  $\length \left(\Pi_{J' \rightarrow H} \circ \Pi_{J' \rightarrow J}\right) \leq \Otil(\gamma_{apxAPSP}) \cdot \length\left(\Pi_{J' \rightarrow H}\right)$, and
	\item\label{lma:CompSpanner:sparsity} $|E(J)|=\widetilde{O}\left(|V(J')|\right)$.
	\end{enumerate}
	 The algorithm runs in time $\widetilde{O}\left(|E(J')| \cdot \length\left(\Pi_{J' \rightarrow H}\right)\cdot \gamma_{apx\textsc{APSP}}\right) + APSP(\Otil(|V(J')|), 3, \Otil(|V(J')|))$.
\end{theorem}

An alternate implementation of a composition spanner based on expander decompositions was given in \cite{maxflow}. Setting $H = J'$ and $\Pi_{J' \rightarrow H} = \text{Id}_{H \rightarrow H}$ immediately yields the following corollary.

\begin{corollary}\label{lma:initspanner}
Given an unweighted, undirected $m$-edge, $n$-vertex graph $G = (V, E)$, there is a deterministic algorithm $\textsc{Sparsify(G)}$, that returns a $\Otil(\gamma_{apxAPSP})$-spanner $H$ of $G$ with $|E(G)| = \Otil(n)$. It further returns an embedding $\Pi_{G \rightarrow H}$ from $H$ to $G$ with $\length(\Pi_{G \rightarrow H}) = \Otil(\gamma_{apxAPSP})$ and $\vcong(\Pi_{G \rightarrow H}) = \Otil(\Delta_{\max}(G))$. The algorithm runs in time $\Otil(m \gamma_{apxAPSP}) + APSP(\Otil(n), 3, \Otil(n))$.
\end{corollary}

\paragraph{Algorithm.} Let $\Delta \defeq \vcong\left(\Pi_{J' \rightarrow H}\right)+\Delta_{\max }(J')$. Then the algorithm consists of layers  $i = 0, \ldots, K \defeq 2 \log(|V(J')|\Delta)$. Each layer $i$ maintains a spanner $J_i$, and a subgraph $\widehat{J}_i \subseteq J_i$, which only contains low congestion edges. 

It considers the edges of $J'$ in arbitrary order, and each edge is first passed to layer $0$. Whenever layer $i$ is passed an edge, it either becomes responsible for the edge, or it passes it on to layer $i + 1$. It then distinguishes three cases. 
\begin{itemize}
    \item If there is a short path $P$ in $\widehat{H}_i$, layer $i$ becomes responsible for the edge and the path $P$ gets added to $\Pi$. Then, if any edge reaches congestion $2 \gamma_{\text{apxAPSP}} \cdot \Delta \cdot \log(n)/2^i$, it is removed from $\widehat{J}_i$ (but not $J_i$). 
    \item Else, if $\forall v \in \Pi_{J' \rightarrow H}(e)$ we have $ \vcong(\Pi_{J_i \rightarrow H}, v) \leq 16 \cdot 2^i \cdot  \length\left(\Pi_{J' \rightarrow H}\right)$ then the edge is added to $J_i$. 
    \item Otherwise, the edge is passed onto layer $i + 1$
\end{itemize}
Finally, our algorithm returns the graph $J = \bigcup_{i = 0}^K J_i$ and embedding $\Pi$. 
See \Cref{algo:sparsify} for detailed pseudocode of our algorithm  

\begin{algorithm}[h]
\caption{\textsc{Sparsify}$\left(H, J', \Pi_{J' \rightarrow H}\right)$}
\label{algo:sparsify}
$\Delta \leftarrow \vcong\left(\Pi_{J' \rightarrow H}\right)+\Delta_{\max }(J')$ \\
$K \leftarrow 2\log(|V(J')| \Delta)$ \\
$J_0, J_1, \ldots, J_K \leftarrow (V, \emptyset)$ \tcp*{Spanners at layers $i = 0, \ldots K$}
$\hat{J}_0, \hat{J}_1, \ldots, \hat{J}_K \leftarrow (V, \emptyset)$ \tcp*{Non congested subgraphs of the spanners}
For every $i = 0, \ldots, K$, initialize \textsc{APSP} datastructure $\mathcal{D}_i$ on $\hat{J}_i$. \\
For all $v \in V(J')$ and $i = 0, \ldots, K$, maintain $\operatorname*{vcong}(\Pi_{J_i \rightarrow H}, v)$. \\
\ForEach{$e = (u, v)  \in E(J')$}{ \label{alg:spanner:main_loop}
	$i \leftarrow 0 $ \\
	\While{$\Pi(e) = \emptyset$} {
	\If{$\mathcal{D}_i.\textsc{Dist}(u, v) \leq 2 \cdot \gamma_{\text {apx\textsc{APSP}}}\cdot \log(n) $}{
		$P \leftarrow \mathcal{D}_i.\textsc{Path}(u, v)$;
		$\Pi(e) \leftarrow P$ \\
        \ForEach{$e' \in P$ s.t. $\econg(\Pi, e') \geq 2\gamma_{\text {apx\textsc{APSP} }} \cdot \Delta \cdot  \log(n) / 2^i $}{ \label{alg:spanner:del_edges}
			Delete $e'$ from $\hat{J}_i$ and from the data structure $\mathcal{D}_i$. 
		}
	} 
    
	\ElseIf{$\forall v \in \Pi_{J' \rightarrow H}(e): \operatorname*{vcong}(\Pi_{J_i \rightarrow H}, v) \leq 32 \cdot 2^i \cdot  \length\left(\Pi_{J' \rightarrow H}\right)$}{\label{alg:spanner:vcong_if}
		Add $e$ to $J_i$ and $\hat{J}_i$. Update $\mathcal{D}_i$ accordingly.  \\
		$\Pi(e) \leftarrow e$;
	}\Else{
        $i \leftarrow i+1$; \\
    }
    }
}
$J \leftarrow \bigcup_{i=0}^K J_i$ \\
\Return{$J$, $\Pi$}
\end{algorithm}

\begin{remark}
    In the setting of \Cref{lma:initspanner}, where $J' = H$ and $\Pi_{J' \rightarrow H} = \operatorname{Id}$, the condition in \Cref{alg:spanner:vcong_if} simply checks if both endpoints of the edge $e = (u,v)$ have degree at most $32 \cdot 2^i$ in $J_i$. 
\end{remark}

\paragraph{Analysis.}
We first define what it means for an edge to be seen. 

\begin{definition}
    For every edge $e \in E(J)$, we say that it was seen by layer $j$ if the computation performed inside the for loop at \Cref{alg:spanner:main_loop} reaches $i \geq j$ when considering edge $e$. 
    Furthermore, we say layer $i$ is responsible for edge $e$. 
\end{definition}

\begin{claim}\label{lma:spannerlvl}If the total number of edges ever seen by layer $i$ is at most $\Delta |V(J')| / 2^i$, then $\left|E(J_i)\right| \leq 6 |V(J')|$ for all $i=0, \ldots, K$.
\end{claim}
\begin{proof}
We start by analysing the number of edges in $J_i \backslash \widehat{J}_i$. For every $e \in E(J_i \backslash \widehat{J}_i)$, we know that it occurs on $2 \gamma_{\text {apx\textsc{APSP} }} \cdot \Delta \cdot \log n / 2^i$ embedding paths. Since the total number of edges seen by layer $i$ is at most $\Delta |V(J')| / 2^i$ and the embedding paths have length at most $2 \gamma_{\text {apx\textsc{APSP} }} \log n$, the total number of edges in $J_i \backslash \widehat{J}_i$ can be at most 

\[\frac{\Delta |V(J')| / 2^i \cdot 2 \gamma_{\text {apx\textsc{APSP} }} \log n}{2 \gamma_{\text {apx\textsc{APSP} }} \cdot \Delta \cdot \log n / 2^i} = |V(J')|
\]
By standard techniques for analysing spanners \cite{althofer1993sparse} we have that the number of edges in $\widehat{J}_i$ is at most $4 |V(J')|$ because it has girth $> 2 \log (|V(J')|)$ by construction.
\end{proof}

\begin{claim}\label{lma:spannerasmpt}Assume the total number ever seen by layer $i$ is $\Delta |V(J')| / 2^i$. Then the total number of edges ever seen by layer $i+1$ is at most $\Delta |V(J')| / 2^{i+1}$.
\end{claim}
\begin{proof} Consider the final spanner $J_i$ computed by layer $i$. For any vertex $v \in V(J')$, recall that $\operatorname*{vcong}(\Pi_{J_i \rightarrow H}, v)$ counts the total number of occurences of $v$ in any of the paths in $\Pi_{J_i \rightarrow H}$. Let $S_i \defeq \{v \in V(J'): \operatorname*{vcong}(\Pi_{J_i \rightarrow H}, v) \geq 32 \cdot \length\left(\Pi_{J' \rightarrow H}\right) \cdot 2^i\}$.  By \Cref{lma:spannerlvl} the total number of edges in $J_i$ is at most $6 |V(J')|$.  Therefore, we have that $\sum_{v \in V(J')} \operatorname*{vcong}(\Pi_{J_i \rightarrow H}, v) \leq 6 |V(J')| \cdot 2\length\left(\Pi_{H \rightarrow J'}\right) = 12 |V(J')| \length\left(\Pi_{H \rightarrow J'}\right)$, as the total amount any $e \in E(J_i)$ can add to the sum is at most $2 \length\left(\Pi_{H \rightarrow J'}\right)$. But this means that $|S_i| \leq |V(J')| / 2^{i+1}$.

By \Cref{alg:spanner:vcong_if}, we only pass an edge $e$ to layer $i+1$ if there exists $v \in S_i$ such that $v \in \Pi_{H \rightarrow J'}(e)$. As at most $ \vcong\left(\Pi_{H \rightarrow J'}\right) \leq \Delta$ many edges can embed through any vertex, the total number of edges passed down is at most $|S_i| \cdot \Delta = |V(J')| \Delta / 2^{i+1}$.
\end{proof}

\begin{lemma}\label{lma:sparsitygood}For all $i=0, \ldots, K$ we have
\begin{enumerate}
	\item $\left|E(J_i)\right| \leq 6 |V(J')|$ and
	\item the number of edges seen by layer $i$ is at most $\Delta |V(J')| / 2^i$.
\end{enumerate}
\end{lemma}
\begin{proof} We proceed by induction. Consider the base case $i=0$. Point $2.$ is true because the total number of edges inserted and thus passed to layer $0$ is at most $\Delta_{\max}(J') |V(J')| \leq \Delta |V(J')|$. The first point is true because of \Cref{lma:spannerlvl}. We can now assume that the claim hold for layer $i$. Then, we have that at most $|V(J')| \Delta / 2^{i+1}$ edges get passed to layer $i+1$ by \Cref{lma:spannerasmpt} which shows the second point. Point $1.$ then follows from \Cref{lma:spannerlvl}.
\end{proof}

\begin{claim}
At the end of the algorithm, we have, for every $v \in V(J')$ and $i = 0, 1, \ldots,K$  that $\operatorname*{vcong}(\Pi_{J_i \rightarrow H}, v) \leq 34 \cdot 2^i \length\left(\Pi_{J' \rightarrow H} \right).$
\end{claim}
\begin{proof}
    Once $\operatorname*{vcong}(\Pi_{J_i \rightarrow H}, v) > 32 \cdot 2^i \cdot \length\left(\Pi_{J' \rightarrow H} \right)$, it can never increase again. But the last time it increases it can do so by at most $\length\left(\Pi_{J' \rightarrow H} \right)$, so at the end $\operatorname*{vcong}(\Pi_{J_i \rightarrow H}, v) \leq 33 \cdot 2^i \cdot \length\left(\Pi_{J' \rightarrow H} \right)$.
\end{proof}

\begin{lemma}\label{lma:vconggood}
	We have $\vcong\left(\Pi_{J \rightarrow H} \circ \Pi_{J \rightarrow \widetilde{J}}\right) \leq \gamma_c \cdot  \length\left(\Pi_{J' \rightarrow H}\right) \left(\vcong\left(\Pi_{J' \rightarrow H}\right)+\Delta_{\max }(J')\right)$,
	\\where $\gamma_c = \Otil(\gamma_{\text {apx\textsc{APSP}}})$
\end{lemma} 
\begin{proof}
	let $J_i$ be the final spanner computed by layer $i$. Consider a fixed vertex $v$. Note that any edge $e \in E(J_i)$ contributes to $\vcong\left(\Pi_{J \rightarrow H} \circ \Pi_{J \rightarrow \widetilde{J}}, v\right)$ by exactly $\econg(\Pi_{J_i \rightarrow J'}, e)$ times the number of times that $v$ occurs on the embedding path $\Pi_{J \rightarrow H}(e)$. This means that the total amount layer $i$ contributes to $\vcong\left(\Pi_{J \rightarrow H} \circ \Pi_{J \rightarrow \widetilde{J}}, v\right)$ is at most $\operatorname*{econg}(\Pi_{J_i \rightarrow J'}) \cdot \operatorname*{vcong}(\Pi_{J_i \rightarrow H}, v) \leq \left(2  \Delta \log(n) \gamma_{\text {apxAPSP }} / 2^i\right) \cdot \big(34 \cdot 2^i  \length\left(\Pi_{J' \rightarrow H}\right)\big) = \Otil(\Delta \length\left(\Pi_{J' \rightarrow H}\right))$. 
    Finally, the number of layers $K = 2 \log(|V(J')| \Delta) = \Otil(1)$ is suitably bounded, which finishes the proof. 
\end{proof}

\begin{proof}[Proof of correctness and runtime]
We start with the correctness analysis. We first note that by \Cref{lma:sparsitygood}, we have that the number of edges seen by layer $K$ is at most $|V(J')| \Delta / 2^K = 1/2$, which proves that every edge is indeed embedded and that the algorithm terminates. 
\begin{enumerate}
    \item This was already shown in \Cref{lma:vconggood}.
    \item By design of the algorithm, we have $\length \left(\Pi_{J' \rightarrow J}\right) \leq 2 \log n \cdot \gamma_{apx\textsc{APSP}}$. Hence we have $\length \left(\Pi_{J' \rightarrow H} \circ \Pi_{J' \rightarrow J}\right) \leq \length \left(\Pi_{J' \rightarrow H}\right) \cdot \length \left(\Pi_{J' \rightarrow J}\right) \leq 2 \log n \cdot \gamma_{apx\textsc{APSP}} \cdot \length \left(\Pi_{J' \rightarrow H}\right)$.
    \item By \Cref{lma:sparsitygood}, each $J_i$ has, at the end of the algorithm, at most $6|V(J')|$ edges. As there are $K = \Otil(1)$ many levels, $J$ has in total at most $\Otil(|V(J')|)$ many edges.
\end{enumerate}
It remains to show the bounds on the runtime. We can maintain the numbers $\vcong(\Pi_{J_i \rightarrow H}, v)$ in time $O(\length\left(\Pi_{J' \rightarrow H}\right))$ per edge that gets passed to $H_i$, yielding a total cost of $\Tilde{O}(|E(J)| \cdot \length\left(\Pi_{J' \rightarrow H}\right))$. Initializing the graphs $J_i$, $\hat{J}_i$ can be done in time $\Otil(n)$.

Consider now the \textsc{APSP} datastructures $\mathcal{D}_i$. By the discussion at the end of \Cref{fact:BSTred}, we can run $\mathcal{D}_i$ on a modified graph with maximum degree at most $3$. Then, as by \Cref{lma:sparsitygood} we have $|E(J_i)| \leq 6 |V(J')|$, this modified graph only has $\Otil(|V(J')|)$ edges and vertices. By initializing one big APSP data structure that will be shared by all layers, we can bound the total costs required by initialization and updates as $APSP(\Otil(|V(J')|, 3, \Otil(|V(J')|))$.

The total number of distance queries can be bounded by $K \cdot |E(J')|$, which thus costs $\Otil(|E(J)|)$ in total. We also need to ask $O(K m) = \Otil(m)$ times for paths of length $\leq 2 \cdot \gamma_{apx\textsc{APSP}} \cdot \log n$, yielding a cost of $\Otil(m \cdot \gamma_{apx\textsc{APSP}})$. The total costs incurred by  \Cref{alg:spanner:del_edges} and \Cref{alg:spanner:vcong_if} can be bounded by $\Otil(m \gamma_{apx\textsc{APSP}})$ and $\Otil(m \cdot \length \left(\Pi_{J \rightarrow J'}\right)))$, respectively, which finishes the runtime analysis. 
\end{proof}

\subsection{Processing a Single Batch of Updates}
\label{sec:single_batch}

In this section, we show that \Cref{lma:CompSpanner} can be used to fix a spanner after a batch of updates have been made with a small loss in vertex congestion and length in our embeddings. As long as the sequential depth of applying this routine is low, this yields a good spanner.

We denote by $U_G$ a batch of updates made to $G$, i.e. a set of updates consisting of edge insertions/deletions, vertex splits and isolated vertex insertions that are applied to $G$. We first state the main lemma of this section. 

\begin{lemma}[Process Batch] \label{lma:procces_batch}
    Given an $n$-vertex graph $G = (V, E)$ and a spanner $H$ with embedding $\Pi_{G \rightarrow H}$, there is an algorithm $\textsc{ProcessBatch}(\cdot)$ that, after an update batch $U_G$ is applied to $G$ can output a batch of updates $U_H$ such that if $H', G'$ are the graphs obtained when applying the respective update batches
    \begin{itemize}
        \item $H' \subseteq G$ and $|U_H| = \Otil(|U_G|)$
        \item it returns a modified embedding $\Pi_{G' \rightarrow H'}$ such that 
    	\begin{enumerate}
            \item $\vcong\left(\Pi_{G' \rightarrow H'}\right) \leq \Otil(\gamma_{apxAPSP}) \cdot \length\left(\Pi_{G \rightarrow H}\right) \cdot \vcong\left(\Pi_{G \rightarrow H}\right) $, and
    	    \item  $\length \left(\Pi_{G' \rightarrow H'}\right) \leq \Otil(\gamma_{apxAPSP}) \cdot \length\left(\Pi_{G \rightarrow H}\right)$.
    	\end{enumerate}
    \end{itemize}
    The algorithm runs in time $\Otil(|U_G| \cdot \gamma_{apxAPSP} \cdot \vcong\left(\Pi_{G \rightarrow H}\right)) + APSP(\Otil(|U_G|), 3, \Otil(|U_G|))$.
\end{lemma}

We remark that both the processing time and number of edges $U_H$ added to $H$ are roughly proportional to the size of the update $U_G$, which nicely sets us up for a batching scheme which we present in \Cref{sec:batching_scheme}.

\paragraph{Algorithm. } Before we describe our algorithm, we define the set of vertices that are touched by an update batch and the edge embedding projection. 

\begin{definition}[Touched Vertices]\label{def:touchedVertex}
We say that an update to $G$ \emph{touches} a vertex $u$ if the update is an edge deletion and $u$ is one of its endpoints, or if the update is a vertex split and $u$ is one of the resulting vertices from the split.
\end{definition}

\begin{definition}[Edge-Embedding Projection]\label{def:edgeProj}
Let $S$ be the set of all vertices touched by updates in $U_G$ and $e =(u, v) \in E(G)$ such that $\Pi_{G \rightarrow H}(e) \cap S \neq \emptyset$. Then, we let $\proj(e)$ be a new edge $\hat{e}$ that is associated with $e$ and has endpoints $\hat{u}$ and $\hat{v}$ being the closest vertices in $S$ to the endpoints of $e$ in $G'$.
\end{definition}

Intuitively, $\hat{u}$ may be thought of as the earliest failure point that our old embedding path $\Pi_{G \rightarrow H}(e)$ from $u$ to $v$ faces once the updates in $U_G$ are forwarded to $H$. In particular, the path segement $\Pi_{G \rightarrow H}(e)[u, \hat{u}]$ remains a valid path. The same holds for $v$ and $\hat{v}$. Hence, if we can update $H$ to find a short path $P_{\hat{e}}$ between $\hat{u}$ and $\hat{v}$, the patched path $\Pi_{G \rightarrow H}(e)[u, \hat{u}] \pconcat P_{\hat{e}} \pconcat \Pi_{G \rightarrow H}(e)[\hat{v}, v]$ will be a valid and short path between $u$ and $v$ again. 

By simply adding the edge $e = (u, v) \in E(G)$ to $H$, we could set $P_{\hat{e}} = \Pi_{G \rightarrow H}(e)[\hat{u}, u] \pconcat e \pconcat \Pi_{G \rightarrow H}(e)[v, \hat{v}]$. Of course, without further steps this does not help at all: by adding the edge $e$, we could simply directly embed it onto itself, and adding an edge for every broken path would lead to very high recourse for $H$.

To circumvent this issue, our algorithm first collects all edges whose embedding was broken in the set $E_{\mathit{affected}}$, and lets $\hat{E}_{\mathit{affected}} = \{\proj(e): e \in E_{\mathit{affected}}\}$. It then considers the graph $J = (S, \hat{E}_{\mathit{affected}})$ with embeddings $\Pi_{J \rightarrow H \cup E_{\mathit{affected}}}$ as described before. Again, for this embedding to be valid we had to add all edges in $E_{\mathit{affected}}$ to $H$, which are way too many. 

Crucially, we have that $|V(J)| = |S| \leq 2|U_G|$, so that we can drastically reduce the number of edges we have to add to $H$. We achieve this by first calling $\textsc{Sparsify}(H \cup E_{\mathit{affected}}, J, \Pi_{J \to H \cup E_{\mathit{affected}}})$, obtaining an edge-sparsified graph $\tilde{J}$ and an embedding $\Pi_{J \rightarrow \tilde{J}}$ with the guarantees from \Cref{lma:CompSpanner}. Then, for every edge $\hat{e} = (\hat{u}, \hat{v}) \in \tilde{J}$, we add the back-projected edge $e = (u, v) \in E(G)$ to $H$,  obtaining the graph $H'$. For all edges $e \in E_{\mathit{affected}}$ where $\proj(e) \not \in E(\Tilde{J})$, we update the embedding by using the $\hat{u}$-$\hat{v}$-path in $H'$ that is given by $[\Pi_{J \to H \cup E_{\mathit{affected}}} \circ \Pi_{J \to\tilde{J}}](\hat{e})$ as a patch. See \Cref{alg:updateSparsifier} for detailed pseudocode.

\begin{algorithm}
Apply updates $U_G$ to $H$; $\Pi_{G' \rightarrow H'} \gets \Pi_{G \rightarrow H}$ \tcp*{new edges in $G'$ map to themselves}
$E_{\mathit{affected}} \gets \{ e \in E(G) \text{ with } \Pi_{G \rightarrow H}(e) \cap S \neq \emptyset\}$ \tcp*{$S$ as in \Cref{def:edgeProj}}
$J \gets (S, \emptyset)$; $\Pi_{G \to H \cup E_{\mathit{affected}}} \gets \emptyset$;
\ForEach{edge $e \in E_{\mathit{affected}}$}{
  $\hat{e} \gets \proj(e)$\tcp*{Find Projected Edge of $e$. See \Cref{def:edgeProj}.}
      Let $u$ and $v$ be the endpoints of $e$ and $\hat{u}$ and $\hat{v}$ be the endpoints of $\hat{e}$;
    Add $\hat{e}$ to $J$.
    $\Pi_{J \to H \cup E_{\mathit{affected}}}(\hat{e}) \gets \Pi_{G \rightarrow H}(e)[\hat{u}, u] \pconcat e \pconcat \Pi_{G \rightarrow H}(e)[v, \hat{v}]$.
}
\tcp{Sparsify projected graph and translate back to patch $H$.}
$(\tilde{J}, \Pi_{J \to \tilde{J}}) \gets \textsc{Sparsify}(H \cup E_{\mathit{affected}}, J, \Pi_{J \to H \cup E_{\mathit{affected}}})$;

\lForEach{edge $e \in E_{\mathit{affected}}$ and $\hat{e} = \proj(e)\in \tilde{J}$}{Add $e$ to $U$.}
\ForEach{edge $e \in E_{\mathit{affected}}$}{%
    $\hat{e} = \proj(e)$;
    Let $a$ and $b$ be the endpoints of $e$, and $\hat{u}$ and $\hat{v}$ be the endpoints of $\hat{e}$.;
    $\Pi_{G' \rightarrow H'}(e) \gets \Pi_{G \rightarrow H}(e)[u, \hat{u}] \pconcat  [\Pi_{J \to H \cup E_{\mathit{affected}}} \circ \Pi_{J \to\tilde{J}}](\hat{e}) \pconcat \Pi_{G \rightarrow H}(e)[\hat{v}, v]$.
}
\Return $U_H = U \cup U_G$, $\Pi_{G' \rightarrow H'}$
\caption{$\textsc{ProcessBatch}(U_G)$}
\label{alg:updateSparsifier}
\end{algorithm}

\begin{proof}[Proof of \Cref{lma:procces_batch}.]
We first prove that the embedding $\Pi_{G' \rightarrow H'}$ maps each edge in $G'$ to a proper path in $H'$. Since $\Pi_{G' \rightarrow H'}$ is initialized to $\Pi_{G \rightarrow H}$ with identity maps for new edges, every edge that is not in $E_{\mathit{affected}}$ has a proper embedding that does not get updated. It remains to show that $\Pi_{G \rightarrow H}(e)[u, \hat{u}] \pconcat  [\Pi_{J \to H \cup E_{\mathit{affected}}} \circ \Pi_{J \to\tilde{J}}](\hat{e}) \pconcat \Pi_{G \rightarrow H}(e)[\hat{v}, v]$ is a proper embedding path. Firstly, $\Pi_{G \rightarrow H}(e)[u, \hat{u}]$ and $\Pi_{G \rightarrow H}(e)[\hat{v}, v]$ are proper embedding paths by \Cref{def:touchedVertex}, and for the same reason $\Pi_{J \rightarrow H \cup E_{\textit{affected}}}(\hat{e})$ is a valid embedding path in $H \cup E_{\textit{affected}}$. By \Cref{lma:CompSpanner}, $[\Pi_{J \to H \cup E_{\mathit{affected}}} \circ \Pi_{J \to\tilde{J}}](\hat{e})$ is then also a proper embedding path from $J$ to $H \cup E_{\textit{affected}}$. Now note that for edges in $E(\Tilde{J})$, we have that the paths $\Pi_{J \rightarrow H \cup E_{\textit{affected}}}$ in fact only use edges in $H'$. Hence $[\Pi_{J \to H \cup E_{\mathit{affected}}} \circ \Pi_{J \to\tilde{J}}](\hat{e})$ is a valid $(\hat{u}, \hat{v})$ path in $H' \subset H \cup E_{\textit{affected}}$.

Next, we show the vertex congestion bound. As the new embedding either consists of unchanged paths or of patched paths of the form $\Pi_{G \rightarrow H}(e)[u, \hat{u}] \pconcat  [\Pi_{J \to H \cup E_{\mathit{affected}}} \circ \Pi_{J \to\tilde{J}}](\hat{e}) \pconcat \Pi_{G \rightarrow H}(e)[\hat{v}, v]$, we can bound $\vcong(\Pi_{G' \rightarrow H'}) \leq 2 \vcong(\Pi_{G \rightarrow H}) + \vcong(\Pi_{J \rightarrow H'} \circ \Pi_{J \to\tilde{J}}) \leq \Otil(\gamma_{apxAPSP}) \cdot \length(\Pi_{G \rightarrow H}) \cdot (\vcong(\Pi_{G \rightarrow H})+\Delta_{\max}(J))$, where the last inequality follows from \Cref{lma:CompSpanner:vcong} in \Cref{lma:CompSpanner}. As by definition of $J$ we have $\Delta_{\max}(J) \leq \vcong(\Pi_{G \rightarrow H})$, the bound follows.

An analogous argument using \Cref{lma:CompSpanner:length} of \Cref{lma:CompSpanner} yields the bound on $\length(\Pi_{G' \rightarrow H'})$.

The runtime bound follow directly from \Cref{lma:CompSpanner} and the definition of our algorithm because $|E_{\mathit{affected}}| \leq \vcong\left(\Pi_{G \rightarrow H}\right)|U_G|$. As $|S| \leq 2|U_G|$, by \Cref{lma:CompSpanner:sparsity} of \Cref{lma:CompSpanner}, we then also have $|E(\Tilde{J})| = \Otil(|U_G|)$, so also $|U_H| \leq |U_G| + |E(\Tilde{J})| = \Otil(|U_G|)$.
\end{proof}

\subsection{Fully Dynamic Spanner}\label{sec:dynamicspanner}
\label{sec:batching_scheme}

Given the $\textsc{ProcessBatch}$ algorithm based on APSP from \Cref{sec:single_batch}, we run the batching scheme from Section 5 of \cite{maxflow} to obtain a dynamic algorithm. We describe the full scheme for clarity and completeness, especially because our $\textsc{ProcessBatch}$ routine depends on bootstrapped APSP instances. 

The algorithm consists of $K + 1$ layers, and layer $i = 0, \ldots, K$ maintains a graph $H_i$, where we initially set $H_0$ to be the spanner from \Cref{lma:initspanner}, and $H_1, \ldots, H_K = (V, \emptyset)$. Throughout $H \defeq \bigcup_{i = 0}^{K} H_i$ denotes the maintained spanner. The algorithm periodically recomputes these graphs, and layers with higher indices are re-computed more often than layers with lower indices. The layers $i$ that do not get re-computed after an update do not change their graph if the update is an insertion. Edge deletions, vertex splits and isolated vertex inserts are however immediately forwarded to $H_i$. 

Additionally, every layer $i$ maintains an embedding $\Pi_i$ that maps a subset of $E(G)$ to $H_{\leq i} \defeq \bigcup_{j = 0}^i H_j$. We do not require that each edge in $E(G)$ only gets mapped by one embedding, and when considering the embedding $\Pi_{\leq i}$ we choose the embedding path of the layer with the largest index $j \leq i$ that contains the edge (i.e. the most frequently updated one). The final embedding will then be $\Pi = \Pi_{\leq K}$. We note that some embeddings might not embed edges to proper paths because of edge deletions that happened since the embedding was computed. The algorithm will ensure that the layer with the highest index that embeds the edge always embeds to a proper path.

Whenever layer $j$ is recomputed, $\Pi_j$ is recomputed such that $\Pi_{\leq j}$ is a valid embedding of the whole graph $G$. To do so, we pass the embedding $\Pi_{\leq j - 1}$ and all updates $U_{j-1}$ that happened since $H_{j - 1}$ was recomputed last to layer $j$, and call $\textsc{UpdateBatch}(U_{j - 1}, H_{j - 1}, \Pi_{\leq j - 1})$ as described in \Cref{lma:procces_batch}. We can assume throughout that $n^{1/K}$ is an integer. See \Cref{alg:fullydynamicupdate} for pseudocode.

\begin{algorithm}
    Update all sparsifiers $H_0, H_1, \ldots, H_K$ with the $t$-th update if it applies. \\
    Add update $t$ to the batches $U_0, U_1, \ldots, U_K$. \\
    $j \gets \min\{j' \in \Z_{\geq 0}: t \text{ is divisible by }n^{1-j'/K}\}$ \\
    $H_{j}, \ldots, H_K \gets (V, \emptyset)$; $\Pi_{j}, \ldots, \Pi_K \gets \emptyset$; $U_{j}, \ldots, U_K \gets \emptyset$\label{alg:fullydynamicupdate:removal}\\
    $H_j, \Pi_j \gets \textsc{ProcessBatch}(U_{j-1}, H_{j-1}, \Pi_{\leq j - 1})$\label{alg:fullydynamicupdate:processbatch}
\caption{$\textsc{Update}(t)$}
\label{alg:fullydynamicupdate}
\end{algorithm}

Since the guarantees achieved by \Cref{algo:sparsify} are equivalent to the ones by the corresponding procedure in \cite{maxflow} (and \cite{CKL24, detmax}), the analysis of all but the runtime remains unchanged and we refer to Lemma $5.6$ and Claims $5.7$ and $5.8$ in \cite{maxflow}.

\begin{claim}[See Lemma $5.6$ and Claims $5.7$ and $5.8$ in \cite{maxflow}] \label{clm:spanner_correctness}
    At all times, $\Pi$ constitutes a valid embedding from $H$ into $G$.  Moreover we have that 
    \begin{itemize}
        \item $\length(\Pi) \leq \Otil(\gamma_{apxAPSP})^{O(K)}$, and
        \item $\vcong(\Pi) \leq \Otil(\gamma_{apxAPSP})^{O(K^2)} \cdot \Delta$
    \end{itemize}
\end{claim}

\begin{proof}[Proof of \Cref{thm:EdgeSparsifier}]
    By \Cref{clm:spanner_correctness}, it remains to prove the statements about sparsity, recourse and runtime. We consider these properties separately for each layer $0 \leq i \leq K$.
    
    As $H_0$ is never updated after initialization, it causes no recourse and the runtime of this layer is, by \Cref{lma:initspanner}, $\Otil(m \gamma_{apxAPSP}) + APSP(\Otil(n), 3, \Otil(n))$. We additionally have $|E(H_0)| = \Otil(n)$.

    Now fix $1 \leq i \leq K$ and note that at all times, $|U_{i-1}| = O(n^{1 - (i-1)/K})$. By \Cref{lma:procces_batch}, we then have that at all times $|E(H_i)| = \Otil(n^{1 - (i-1)/K})$. This immediately implies the sparsity of $H$. Moreover, whenever layer $i$ is recomputed, the recourse caused by the edge removals occurring in \Cref{alg:fullydynamicupdate:removal} is $\Otil(1) \cdot \sum_{j = i}^{K} n^{1 - (j-1)/K} = \Otil(n^{1 - (i-1)/K})$. The runtime costs of that line are proportional to the recourse. When $\textsc{ProcessBatch}(\cdot)$ is called in  \Cref{alg:fullydynamicupdate:processbatch}, its runtime is $\Otil(n^{1-(i-1)/K} \gamma_{apxAPSP} \vcong(\Pi_{\leq i-1}))$ + $APSP(\Otil(n^{1-(i-1)/K}), 3, \Otil(n^{1-(i-1)/K}))$, and the recourse caused (solely through edge insertions) is $\Otil(n^{1-(i-1)/K})$. From \Cref{clm:spanner_correctness}, we know that at all times we have $\vcong(\Pi_{\leq i-1}) \leq \Otil(\gamma_{apxAPSP})^{O(K^2)} \cdot \Delta$. Hence, the re-computation of layer $i$ can be done in time $n^{1-(i-1)/K} \cdot \Otil(\gamma_{apxAPSP})^{O(K^2)} \Delta + APSP(\Otil(n^{1-(i-1)/K}), 3, \Otil(n^{1-(i-1)/K}))$ with recourse $\Otil(n^{1-(i-1)/K})$.

    Noting that layer $i$ gets re-computed at most $n/n^{i/K} = n^{1-i/K}$ many times, and then summing the total costs of all layers, we can bound the total run time of the algorithm by
    \begin{align*}
        (K+1) \Otil(\gamma_{apxAPSP})^{O(K^2)} \cdot n^{1+1/K} \cdot \Delta + \sum_{i=0}^{K}n^{i/K}APSP(\Otil(n^{1-(i-1)/K}), 3, \Otil(n^{1-(i-1)/K})).
    \end{align*} We can replace the latter term by $APSP(\Otil(n^{1+1/K}), 3, \Otil(n^{1+1/K}))$ by noting that we can initialize one big APSP data structure at the beginning of the algorithm that all subsequent calls of $\textsc{UpdateBatch}(\cdot)$ will share. As $K = o(\log^{1/3}n)$, this finishes the runtime analysis.

    Finally, after $\norm{G^{(t)}}$ many updates have been made, we know that layer $i$ got called at most $\norm{G^{(t)}} / n^{1-i/K}$ many times, and each time it got called it led to $\Otil(n^{1 - (i-1)/K})$ many updates to our spanner. Hence we can bound the recourse at any time by
    $$
    \norm{H^{(t)}} \leq \sum_{i = 1}^{K} \bigg(\big(\norm{G^{(t)}} / n^{1-i/K}\big) \cdot \big(\Otil(n^{1 - (i-1)/K})\big)\bigg) = \Otil(n^{1/L}) \cdot \norm{G^{(t)}}.
    $$
\end{proof}

\section*{Acknowledgements}

The authors are very grateful for insightful discussions with Maximilian Probst Gutenberg. 

\newpage

\bibliographystyle{alpha}
\bibliography{refs}

\newpage

\appendix

\section{Dynamic Vertex Sparsifiers via Low-Stretch Trees}
\label{apx:vertex_sparsifier}

In \cite{kyng2023dynamic}, the authors dynamize the vertex sparsifier of \cite{andoni20vert}. Since we use a simpler, white-boxed part of their algorithm, we sketch their proof for completeness. 

We first restate the vertex sparsifer theorem as used in \Cref{sec:vertex_sparsification}. 

\vertexSparsifier*

The vertex sparsifier is constructed in two steps.

Firstly, we compute a path collection $\mathcal{P}$ consisting of $k^{o(1)} \cdot n \cdot \Delta$ paths. This collection remains static throughout the sequence of updates, except that we remove paths that use deleted edges and add inserted edges as paths of length one. Then, our algorithm maintains a pivot set $A$ as described such that the distances between vertices in $A$ are approximately preserved when only using paths in $\mathcal{P}$.   

Secondly, the path collection $\mathcal{P}$ is used to achieve the desired vertex sparsification via core graphs. Core graphs are partial static low stretch spanning trees. Both steps are outlined in more detail below 

\paragraph{The Path Set: Dynamizing the \cite{andoni20vert}-Vertex Sparsifier. } 

We initially let $A$ be a set of $n/k$ vertices such that $|B(v, A)| \leq \tilde{O}(k)$ for all $v \in V$.\footnote{While this can be achieved directly via random sampling, \cite{kyng2023dynamic} derandomize this construction in Appendix A.1.} Then, the path set $\mathcal{P}$ is constructed as follows

\begin{itemize}
    \item For every vertex $u, v \in E$, and every pair of vertices $u', v'$ such that $u' \in \bar{B}(u, A)$ and $v' \in \bar{B}(v, A)$, add the path $\pi_{u', u} \concat (u, v) \concat \pi_{v, v'}$ to $\mathcal{P}$. 
\end{itemize}

Then, whenever an update (i.e. an edge insertion or deletion) occurs, we add both endpoints of the affected edge to $A$. If its a deletion, we remove broken paths from $\mathcal{P}$, and if its an insertion we add it to $\mathcal{P}$ as a path of length one. 

This ensures that the following crucial invariant is preserved. Recall that the pivot $p(v)$ is defined to be the closest vertex to $v$ in $A$. 

\begin{invariant}
    \label{inv:invariant_paths}
    For every edge $(u, v) \in E$, the path $\pi_{p(u), u} \concat (u, v) \concat \pi_{v, p(v)}$ from $u$ to $p(v)$ is in $\mathcal{P}$. 
\end{invariant}
\begin{proof}
    If $(u,v)$ is an inserted edge then $p(u) = u$, $p(v) = v$ and the statement follows. Otherwise, it suffices to notices that $\bar{B}(u, A)$ contains all possible pivots of $u$, and that it is not possible to tamper with the initial shortest path $\pi_{p(u), u}$, as otherwise $p(u)$ is no longer be the pivot of $u$. 
\end{proof}

Given \Cref{inv:invariant_paths}, we conclude that distances between vertices $u, v \in A$ are conserved up to a small constant factor analogously to the static construction of \cite{andoni20vert}. 

\paragraph{Core Graphs: Vertex Reduction and Maintaining Path Lengths. }

A core graph is a forest $F \subseteq G$, and each rooted tree component of the forest is referred to as a core.

For every edge $e = (u,v) \in E$, we define its stretch $\str_F(e)$ as follows 
\begin{align*}
    \str_F(e) \defeq 
    \begin{cases}
        1 + \ll(F[u,v])/\ll(e) & \text{if } \root(u) = \root(v)  \\
        1 + (\ll(F[\root(u), u]) + \ll(F[v, \root(v)]))/\ll(e) & \text{else} 
    \end{cases}
\end{align*}
and after contracting cores, we replace edge lengths by $\ll(e) + \ll(F[\root(u), u]) + \ll(F[v, \root(v)])$. Here, we let $F[u, v]$ be the unique forest path between $u$ and $v$. 

We ensure that every forest component has size at most $k^{o(1)}$, contains at most one vertex in $A$, and that all vertices in $A$ are roots. 

\cite{maxflow} and \cite{detmax} gave a deterministic algorithm for maintaining such a core graph with initial stretch over-estimates $\tilde{\str}(e)$ for each edge such that a constant fraction of the paths in $\mathcal{P}$ only get stretched by $\tilde{O}(1)$ when measured with respect to $\tilde{\str}(e)$, and the stretch $\str(e)$ of edge $e$ remains bounded by $\tilde{\str}(e)$ at all times. A simple multiplicative weights scheme then allows \cite{kyng2023dynamic} to use $\tilde{O}(1)$ core graphs to preserve all paths in some core graph. 

Finally, we add edges of length one between roots that are copies of the same vertices in different core graphs to obtain $\tilde{G}$ as described in the theorem.\footnote{\cite{kyng2023dynamic} collapses such vertices. This is not necessary for our purposes.}

\begin{remark}
    For \Cref{thm:RecourseVertexSparsifier}, we observe that we can first make vertices terminals before we split them. Then, one can check that the stretch of edges does not increase after a split. 
\end{remark}

\end{document}